\documentclass[10pt]{article}
\usepackage{authblk}

\usepackage{fullpage}

\usepackage{amsthm,amssymb,amsfonts,amsmath}
\usepackage{mathrsfs}
\usepackage{graphicx}
\usepackage{color}
\usepackage{stmaryrd} \usepackage[ruled,linesnumbered]{algorithm2e}
\usepackage{./kbordermatrix}
\usepackage[utf8x]{inputenc}

\usepackage{tikz}
\usetikzlibrary{positioning}
\usetikzlibrary{arrows,automata}
\usepackage{caption}
\usepackage{subcaption}
\usepackage[authoryear,round]{natbib}

\newcommand{\R}{\mathbb{R}}
\newcommand{\N}{\mathbb{N}}
\newcommand{\Z}{\mathbb{Z}}

\renewcommand{\P}{\mathbb{P}}

\newcommand{\mat}[1]{\mathbf{#1}}
\newcommand{\rank}{\operatorname{rank}}
\newcommand{\linspan}{\operatorname{span}}
\newcommand{\supp}{\operatorname{supp}}
\newcommand{\diag}{\operatorname{diag}}
\newcommand{\vecm}{\operatorname{vec}}

\newcommand{\norm}[1]{\|#1\|}
\newcommand{\bnorm}[1]{\left\|#1\right\|}
\newcommand{\sval}{\sigma}

\DeclareMathOperator{\Tr}{Tr}
\newcommand{\symbola}{a}

\newcommand{\ltwo}{\ell^2}
\newcommand{\lone}{\ell^1}

\newcommand{\lp}{\ell^p}
\newcommand{\lpq}{\ell^q}
\newcommand{\ltworat}{\ell^2_{\cR}}
\newcommand{\lonerat}{\ell^1_{\cR}}
\newcommand{\lprat}{\ell^p_{\cR}}
\newcommand{\finsup}{\mathcal{C}_{00}}

\newcommand{\sstar}{\Sigma^\star}

\newcommand{\cR}{\mathcal{R}}
\newcommand{\A}{\mat{A}}
\newcommand{\hA}{\hat{\mat{A}}}
\newcommand{\tA}{\tilde{\mat{A}}}
\newcommand{\bA}{\bar{\mat{A}}}
\newcommand{\B}{\mat{B}}

\newcommand{\balpha}{\boldsymbol{\alpha}}
\newcommand{\bbeta}{\boldsymbol{\beta}}
\newcommand{\azero}{\balpha}
\newcommand{\tazero}{\tilde{\balpha}}
\newcommand{\hazero}{\hat{\balpha}}
\newcommand{\bazero}{\bar{\balpha}}
\newcommand{\ainf}{\bbeta}
\newcommand{\hainf}{\hat{\bbeta}}
\newcommand{\tainf}{\tilde{\bbeta}}
\newcommand{\bainf}{\bar{\bbeta}}
\newcommand{\bzero}{\boldsymbol{\beta}_0}
\newcommand{\binf}{\boldsymbol{\beta}_{\infty}}

\newcommand{\wa}{\langle \azero, \ainf, \{\A_\symbola\} \rangle}

\newcommand{\waS}{\langle \azero, \ainf, \{\A_\symbola\}_{\symbola \in \Sigma} \rangle}
\newcommand{\waQ}{\langle \mQ^\top \azero, \mQ^{-1} \ainf, \{\mQ^{-1} \A_\symbola
\mQ\} \rangle}
\newcommand{\hwa}{\langle \hazero, \hainf, \{\hat{\A}_\symbola\} \rangle}
\newcommand{\bwa}{\langle \bazero, \bainf, \{\bA_\symbola\} \rangle}
\newcommand{\twa}{\langle \tazero, \tainf, \{\tilde{\A}_\symbola\} \rangle}

\newcommand{\Ps}{\mathcal{P}}
\newcommand{\Ss}{\mathcal{S}}

\newcommand{\normop}[1]{\|#1\|_{\mathrm{op}}}
\newcommand{\normsp}[1]{\|#1\|_{\mathrm{S},p}}
\newcommand{\normsone}[1]{\|#1\|_{\mathrm{S},1}}
\newcommand{\normstwo}[1]{\|#1\|_{\mathrm{S},2}}
\newcommand{\normsinf}[1]{\|#1\|_{\mathrm{S},\infty}}

\newcommand{\normtr}[1]{\|#1\|_{\mathrm{tr}}}
\newcommand{\normf}[1]{\|#1\|_{\mathrm{F}}}

\renewcommand{\H}{\mat{H}}

\renewcommand{\v}{\mat{v}}
\newcommand{\mM}{\mat{M}}
\newcommand{\mN}{\mat{N}}
\newcommand{\mP}{\mat{P}}
\newcommand{\mS}{\mat{S}}
\newcommand{\mG}{\mat{G}}

\newcommand{\mU}{\mat{U}}
\newcommand{\mV}{\mat{V}}
\newcommand{\mW}{\mat{W}}
\newcommand{\mD}{\mat{D}}
\newcommand{\mQ}{\mat{Q}}
\newcommand{\mR}{\mat{R}}
\newcommand{\mI}{\mat{I}}
\newcommand{\mJ}{\mat{J}}
\newcommand{\mT}{\mat{T}}
\newcommand{\mX}{\mat{X}}
\newcommand{\mY}{\mat{Y}}
\newcommand{\mZ}{\mat{Z}}
\newcommand{\mL}{\mat{L}}
\newcommand{\vx}{\mat{x}}
\newcommand{\vy}{\mat{y}}
\newcommand{\vz}{\mat{z}}

\newcommand{\me}{\mat{e}}
\newcommand{\mPi}{\mat{\Pi}}
\newcommand{\mGamma}{\mat{\Gamma}}
\newcommand{\TrD}{\Tr^{\mD}}

\newtheorem{theorem}{Theorem}[section]
\newtheorem{lemma}[theorem]{Lemma}
\newtheorem{proposition}[theorem]{Proposition}
\newtheorem{corollary}[theorem]{Corollary}
\newtheorem{definition}[theorem]{Definition}

\title{Singular Value Automata and Approximate Minimization}
\author[1]{Borja Balle}
\author[2]{Prakash Panangaden}
\author[2]{Doina Precup}
\affil[1]{Amazon Research\footnote{This work was completed while the author was at Lancaster University.}}
\affil[2]{School of Computer Science, McGill University}

\begin{document}
\maketitle

\begin{abstract}
The present paper uses spectral theory of linear operators to construct approximately minimal realizations of weighted languages.  Our new contributions are: (i) a new algorithm for the SVD decomposition of finite-rank infinite Hankel matrices based on their representation in terms of weighted automata, (ii) a new canonical form for weighted automata arising from the SVD of its corresponding Hankel matrix and (iii) an algorithm to construct approximate minimizations of given weighted automata by truncating the canonical form.  We give bounds on the quality of our approximation.
\end{abstract}

\section{Introduction}

When one considers \emph{quantitative systems} it becomes meaningful to ask about the
\emph{approximate} minimization of transition systems or automata.  This
concept, meaningless for ordinary automata, is appropriate for many types
of systems: weighted automata, probabilistic automata of various kinds, and
timed automata.  The present paper focuses on weighted automata where we
are able to exploit spectral theory of linear operators to construct
approximately minimal realizations of weighted languages.  Our main
contributions are:
\begin{itemize}
\item A new algorithm for the SVD decomposition of finite-rank infinite Hankel matrices
based on their representation in terms of weighted automata (Sections~\ref{sec:fundeqns} and~\ref{sec:computegramians}).

\item A new canonical form for weighted automata arising from the SVD of its
corresponding Hankel matrix (Section~\ref{sec:sva}).

\item An algorithm to construct approximate minimizations of given weighted
automata by truncating the canonical form (Section~\ref{sec:approxmin}).
\end{itemize}

Minimization of automata has been a major subject since the 1950s, starting
with the now classical work of the pioneers of automata theory.  Recently
there has been activity on novel algorithms for minimization based on
duality~\citep{Bezhanishvili12,Bonchi14} which are ultimately based on a
remarkable algorithm due to Brzozowski from the 1960s~\citep{Brzozowski62}.
The general co-algebraic framework permits one to generalize Brzozowski's
algorithm to other classes of automata like weighted automata.

Weighted automata are also used in a variety of practical settings, such as
machine learning where they are used to represent predictive models for
time series data and text.  For example, weighted automata are commonly
used for pattern recognition in sequences occurring in speech
recognition~\citep{MohriPereiraRiley2008}, image
compression~\citep{AlbertLari2009}, natural language
processing~\citep{knight2009applications}, model
checking~\citep{baier2009model}, and machine
translation~\citep{DeGispert2010}.  The machine learning motivations of our
work are discussed at greater length in Section~\ref{sec:related}, as they are the
main impetus for the present work.  There has also been interest in this
type of representations in the general theory of quantitative systems,
including concurrency theory~\citep{boreale2009weighted} and
semantics~\citep{Bonchi2012}.

While the detailed discussion of the machine learning motivations appears
in the related work section, it is appropriate to make a few points at the
outset.  First, the formalism of weighted finite automata (WFA) serves as a
unifying formalism; examples of models that are subsumed include: hidden
Markov models (HMM), predictive representations of state (PSR), and
probabilistic automata of various kinds.  Second, in many learning
scenarios one has to make a guess of the number of states in advance of the
learning process; the resulting algorithm is then trying to construct as
best it can a minimal realization within the given constraint.  Thus our
work gives a general framework for the analysis of these types of learning
scenarios.  

The present paper extends and improves the results of our previous work \citep{balle2015canonical}, where the singular value automaton was defined for the first time.
The contents of this paper are organized as follows.
Section~\ref{sec:background} defines the notation that will be used
throughout the paper and reviews a series of well-known results that will be needed.
Section~\ref{sec:ratfuns} develops some basic results on analytic properties of rational series computed by weighted automata.
Section~\ref{sec:sva} establishes the existence of the singular value automaton, a canonical form for weighted automata computing square-summable rational series.
Section~\ref{sec:fundeqns} proves some fundamental equations satisfied by singular value automata and provides an algorithm for computing the canonical form.
Section~\ref{sec:computegramians} shows how to implement the algorithms from the previous section using two different methods for computing the Gramian matrices associated with a factorization of the Hankel matrix.
Section~\ref{sec:approxmin} describes the main application of singular value automata to approximate minimization and proves an important approximation result.
Section~\ref{sec:related} discusses related work in approximate minimization, spectral learning of weighted automata, and the theory of linear dynamical systems.
We conclude with Section~\ref{sec:conclusion}, where we point out interesting
future research directions.

\section{Notation and Preliminaries}\label{sec:background}

Given a positive integer $d$, we denote $[d] = \{1,\ldots,d\}$.  We use $\R$ to denote the field of real numbers, and $\N = \{0,1,\ldots\}$ for the commutative monoid of natural numbers.  In this section we present notation and preliminary results about linear algebra, functional analysis, and weighted automata that will be used throughout the paper.  We state all our results in terms of real numbers because this is the most common choice in the literature on weighted automata, but all our results remain true (and the proofs are virtually the same) if one considers automata with weights in the field of complex numbers $\mathbb{C}$. 

\subsection{Linear Algebra and Functional Analysis}

We use bold letters to denote vectors $\v \in \R^d$ and matrices $\mM \in
\R^{d_1 \times d_2}$. Unless explicitly stated, all vectors are column vectors.
We write $\mI$ for the identity matrix, $\diag(a_1,\ldots,a_n)$ for a
diagonal matrix with $a_1, \ldots, a_n$ in the diagonal, and
$\diag(\mM_1,\dots,\mM_n)$ for the block-diagonal matrix containing the square
matrices $\mM_i$ along the diagonal.
The $i$th coordinate vector $(0,\ldots,0,1,0,\ldots,0)^\top$ is denoted by
$\me_i$ and the all ones vector $(1,\ldots,1)^\top$ is denoted by
$\mat{1}$. 
For a matrix $\mM \in \R^{d_1 \times d_2}$, $i \in [d_1]$, and $j \in [d_2]$, we
use $\mM(i,:)$ and $\mM(:,j)$ to denote the $i$th row and the $j$th column of
$\mM$ respectively.
Given a matrix $\mM \in \R^{d_1 \times d_2}$ we denote by
$\vecm(\mM) \in \R^{d_1 \cdot d_2}$ the vector obtained by concatenating the columns of
$\mM$ so that $\vecm(\mM)((i-1) d_2 + j) = \mM(i,j)$.
Given two matrices $\mM \in \R^{d_1 \times d_2}$ and $\mM' \in \R^{d_1' \times
d_2'}$ we denote their tensor (or Kronecker) product by $\mM \otimes \mM' \in
\R^{d_1 d_1' \times d_2 d_2'}$, with entries given by $(\mM \otimes
\mM')((i-1)d_1' + i', (j-1)d_2' + j') = \mM(i,j) \mM'(i',j')$, where $i \in
[d_1]$, $j \in [d_2]$, $i' \in [d_1']$, and $j' \in [d_2']$.
For simplicity, we will sometimes write $\mM^{\otimes 2} = \mM \otimes \mM$, and
similarly for vectors.
A \emph{rank factorization} of a rank $n$ matrix $\mM \in \R^{d_1 \times d_2}$
is an expression of the form $\mM = \mQ \mR$ where $\mQ \in \R^{d_1 \times n}$
and $\mR \in \R^{n \times d_2}$ are full-rank matrices; i.e.\ $\rank(\mQ) = \rank(\mR) = \rank(\mM) = n$.
When $\mQ$ is a square invertible matrix, we use the shorthand notation $\mQ^{-\top}$ to denote the transpose of its inverse $(\mQ^{-1})^\top$.

Given a matrix $\mM \in \R^{d_1 \times d_2}$ of rank $n$, its \emph{singular
value decomposition} (SVD)\footnote{To be more precise, this is a \emph{compact}
singular value decomposition, since the inner dimensions of the decomposition
are all equal to the rank. In this paper we shall always use the term SVD to
mean compact SVD.} is a decomposition of the form $\mM = \mU \mD \mV^\top$ where
$\mU \in \R^{d_1 \times n}$, $\mD \in \R^{n \times n}$, and $\mV \in \R^{d_2
\times n}$ are such that: $\mU^\top \mU = \mV^\top \mV = \mI$, and $\mD =
\diag(\sval_1,\ldots,\sval_n)$ with $\sval_1 \geq \cdots \geq \sval_n > 0$. The
columns of $\mU$ and $\mV$ are thus orthonormal and are called left and right \emph{singular vectors}
respectively, and the $\sval_i$ are its \emph{singular values}.
The SVD is unique (up to sign changes in associate singular vectors) whenever
all inequalities between singular values are strict.
The Moore--Penrose pseudo-inverse of $\mM$ is denoted by $\mM^\dagger$ and
is the \emph{unique} matrix (if it exists) such that
$\mM \mM^\dagger \mM = \mM$ and $\mM^\dagger \mM \mM^\dagger = \mM^\dagger$.  It
can be computed from the SVD $\mM = \mU \mD \mV^\top$ as
$\mM^\dagger = \mV \mD^{-1} \mU^\top$.

For $1 \leq p \leq \infty$ we will write $\norm{\v}_p$ for the $\lp$ norm of
vector $\v$.
The corresponding \emph{induced norm} on matrices is $\norm{\mM}_p =
\sup_{\norm{\v}_p = 1} \norm{\mM \v}_p$.
We recall the following characterizations for induced norms with $p \in \{1, \infty\}$: $\norm{\mM}_1 = \max_j \sum_{i} |\mM(i,j)|$ and $\norm{\mM}_{\infty} = \max_i \sum_{j} |\mM(i,j)|$.
In addition to induced norms we will also use Schatten norms.
If $\mM$ is a rank-$n$ matrix with singular values $\mat{s} = (\sval_1, \ldots,
\sval_n)$, the \emph{Schatten $p$-norm} of $\mM$ is given by $\normsp{\mM} =
\norm{\mat{s}}_p$.
Most of these norms have given names: $\norm{\cdot}_2 = \normsinf{\cdot} =
\normop{\cdot}$ is the \emph{operator (or spectral) norm}; $\normstwo{\cdot} =
\normf{\cdot}$ is the \emph{Frobenius norm}; and $\normsone{\cdot} =
\normtr{\cdot}$ is the \emph{trace (or nuclear) norm}.
For a matrix $\mM$ the \emph{spectral radius} is the largest modulus $\rho(\mM)
= \max_i |\lambda_i(\mM)|$ among the eigenvalues $\lambda_i(\mM)$ of $\mM$.
For a square matrix $\mM$, the series $\sum_{k \geq 0} \mM^k$ converges if and
only if $\rho(\mM) < 1$, in which case the sum yields $(\mI - \mM)^{-1}$.

Recall that if a square matrix $\mM \in \R^{d \times d}$ is symmetric then
all its eigenvalues are real. A symmetric matrix $\mM$ is \emph{positive
  semi-definite} when all its eigenvalues are non-negative; we denote this
fact by writing $\mM \geq \mat{0}$, where $\mat{0}$ is a zero $d \times d$ matrix.  The Loewner partial ordering on the set of all
$d \times d$ matrices is obtained by defining $\mM_1 \geq \mM_2$ to
mean $\mM_1 - \mM_2 \geq \mat{0}$.  The fact that this gives a partial
order follows from the fact that the positive semi-definite operators form
a cone.  In particular, $\mM_1 \geq \mM_2$ implies the trace inequality $\Tr(\mM_1) \geq \Tr(\mM_2)$.

Sometimes we will name the columns and rows of a matrix using
ordered index sets $\mathcal{I}$ and $\mathcal{J}$. In this case we will
write $\mat{M} \in \R^{\mathcal{I} \times \mathcal{J}}$ to denote a matrix of
size $|\mathcal{I}| \times |\mathcal{J}|$ with rows indexed by $\mathcal{I}$ and
columns indexed by $\mathcal{J}$.

Recall that a \emph{Banach space} is a complete normed vector space
$(X, \norm{\cdot}_X)$. A \emph{Hilbert space} is a Banach space
$(X, \norm{\cdot}_X)$ where the norm arises from an inner product:
$\norm{x}_X^2 = \langle x,x \rangle_X$. A Hilbert space is separable if it
admits a countable orthonormal basis. The \emph{operator norm} of a linear
operator $T : X \to Y$ between two Banach spaces is given by
$\normop{T} = \sup_{\norm{x}_X \leq 1} \norm{T x}_Y$. The operator is
\emph{bounded} (and continuous) if $\normop{T}$ is finite. An operator
$T : X \to Y$ is \emph{compact} if the closure in the topology of $Y$ of
the image under $T$ of the unit ball in $X$ is a compact set in $Y$. A
sufficient condition for compactness is to be a bounded finite-rank
operator.

Our main interest in compact operators is motivated by the existence of a
decomposition equivalent to SVD for compact operators in Hilbert
spaces. Note that for a bounded operator $T : X \to Y$ between separable
Hilbert spaces it is possible to choose countable orthonormal basis $F =
(f_j)_{j \in \mathcal{J}}$ and $E = (e_i)_{i \in \mathcal{I}}$ for $X$ and
$Y$ respectively, and write down an infinite matrix $\mT \in
\R^{\mathcal{I} \times \mathcal{J}}$ for $T$ with entries given by
$\mT(i,j) = \langle e_i, T f_j \rangle_Y$. In the case of finite-rank
bounded operators the \emph{Hilbert--Schmidt decomposition}
\citep{zhu1990operator} provides a decomposition for the infinite matrix
associated with an operator analogous to the compact SVD decomposition for
finite matrices. In particular, if $T$ has rank $n$, then the decomposition
theorem yields singular values $\sigma_1 \geq \cdots \geq \sigma_n > 0$ and
singular vectors $v_i \in X$ and $u_i \in Y$ for $i \in [n]$ such that for
all $x \in X$ we have 
\begin{equation}
Tx = \sum_{i=1}^n \sigma_i \langle v_i, x \rangle_X u_i \enspace.
\end{equation}
By writing the singular vectors $u_i$ and $v_i$ in terms of the bases $E$
and $F$ we can write this decomposition as $\mT = \mU \mD \mV^\top$ with
$\mU \in \R^{\mathcal{I} \times n}$ and $\mV \in \R^{\mathcal{J} \times n}$
satisfying the same properties as the SVD for finite matrices. 

\subsection{Weighted Automata and Rational Functions}

Let $\Sigma$ be a fixed finite alphabet with $|\Sigma| < \infty$ symbols, and
$\sstar$ the set of all finite strings with symbols in $\Sigma$. We use
$\varepsilon$ to denote the empty string.
Given two strings $p, s \in \sstar$ we write $w = p s$ for their concatenation,
in which case we say that $p$ is a prefix of $w$ and $s$ is a suffix of $w$.
We denote by $|w|$  the length (number of symbols) in a string $w \in
\sstar$.
Given a set of strings $X \subseteq \sstar$ and a function $f : \sstar \to \R$,
we denote by $f(X)$ the summation $\sum_{x \in X} f(x)$ if defined. For example,
we will write $f(\Sigma^t) = \sum_{|x| = t} f(x)$ for any $t \geq 0$. The notation $\Sigma^{< t}$ (resp.\ $\Sigma^{\leq t}$) denotes all string of length less than (resp.\ at most) $t$. As customary, we use $\Sigma^+$ to denote the set of non-empty strings.

Now we introduce our notation for weighted automata. We want to note that we
will not work with weights in arbitrary semi-rings; this paper only
considers automata with real weights and the usual addition and multiplication
operations. In addition, instead of resorting to the usual description of
automata as directed graphs with labelled nodes and edges, we will use a
linear-algebraic representation which is more convenient for our purposes.
Thus, a \emph{weighted finite automata} (WFA) of dimension $n$ over $\Sigma$ is a
tuple $A = \waS$ where $\azero \in \R^n$ is the vector of \emph{initial
weights}, $\ainf \in \R^n$ is the vector of \emph{final weights}, and for each
symbol $a \in \Sigma$ the matrix $\A_a \in \R^{n \times n}$ contains
the \emph{transition weights} associated with $a$. An example is provided in Figure~\ref{fig:wfa}.
Note that in this representation a fixed initial state is given by $\azero$ (as
opposed to formalisms that only specify a transition structure), and the
transition endomorphisms $\A_a$ and the final linear form $\ainf$ are given
in a fixed basis on $\R^n$ (as opposed to abstract descriptions where these
objects are represented as basis-independent elements objects on an abstract $n$-dimensional
vector space).

\begin{figure}[t]
\centering
\begin{subfigure}[c]{0.4\textwidth}
\centering
\begin{tikzpicture}[shorten >=1pt,node distance=3cm,auto]
\tikzstyle{every state}=[draw=blue!50,very thick,on grid,fill=blue!20,inner sep=3pt,minimum size=4mm]
\node[state with output,initial,initial text=$1$] (q_0) {$q_1$ \nodepart{lower} $1$};
\node[state with output,initial,initial text=$-2$] (q_1) [right of=q_0] {$q_2$ \nodepart{lower} $-1$};
\path[->]
(q_0) edge [loop above] node
{\footnotesize $\begin{matrix}{{a , 1}}\\b , 0 \end{matrix}$} ()
(q_0) edge [bend right] node [swap]
{\footnotesize $\begin{matrix}{{a , -1}}\\{b , -2}\end{matrix}$} (q_1)
(q_1) edge [loop above] node
{\footnotesize $\begin{matrix}{{a , 3}}\\b , 5 \end{matrix}$} ()
(q_1) edge [bend right] node [swap]
{\footnotesize $\begin{matrix}{{a , -2}}\\{b , 0}\end{matrix}$} (q_0);
\end{tikzpicture}\caption{\label{fig:trans}}
\end{subfigure}
\begin{subfigure}[c]{0.4\textwidth}
\centering
\begin{minipage}[t]{\textwidth}
\begin{tabular}{cc}
$\azero =
\left[
\begin{matrix}
1\\
-2\\
\end{matrix}
\right]$
&
$\A_a =
\left[
\begin{matrix}
1 & -1 \\
-2 & 3\\
\end{matrix}
\right]$\\[.5cm]
$\ainf =
\left[
\begin{matrix}
1\\
-1\\
\end{matrix}
\right]$
&
$\A_b =
\left[
\begin{matrix}
0 & -2 \\
0 & 5\\
\end{matrix}
\right]$
\end{tabular}
\end{minipage}
\caption{\label{fig:weights}}
\end{subfigure}
\caption{(\protect\subref{fig:trans}) Example of WFA $A$ with two
  states. Within each circle we denote the state number $q_i$ and the
  corresponding final weight. The initial weights are denoted using arrows
  pointing to each state, and the transition weights are given by arrows
  between states. For example, $f_A(ba) = 1 \times (-2) \times 3 \times
  (-1) + 1 \times (-2) \times (-2) \times 1 + (-2) \times 5 \times 3 \times
  (-1) + (-2) \times 5 \times (-2) \times 1 = 60$. (\protect\subref{fig:weights}) Corresponding initial vector $\azero$, final vector $\ainf$, and transition matrices $\A_a$ and $\A_b$.}
\label{fig:wfa} 
\end{figure}
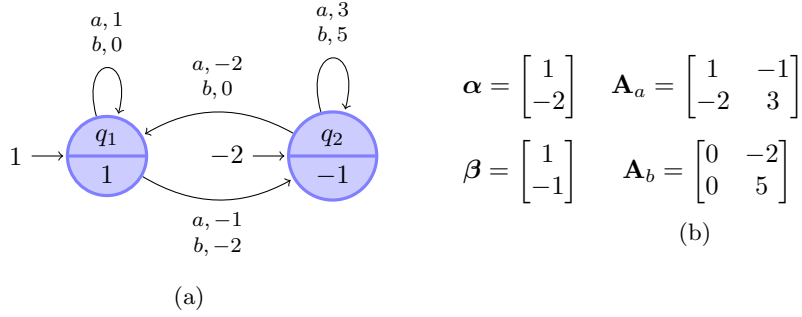

We will use $\dim(A)$ to denote the dimension of a WFA, to which we sometimes also refer to as the number of states in the WFA. The state-space of a WFA of dimension $n$ is identified with the integer set $[n]$.
Every WFA $A$ \emph{realizes} a function $f_A : \sstar \to \R$ which, given a
string $x = x_1 \cdots x_t \in \sstar$, produces
\begin{equation*}
f_A(x) = \azero^\top \A_{x_1} \cdots \A_{x_t} \ainf =
\azero^\top \A_x \ainf \enspace,
\end{equation*}
where we defined the shorthand notation $\A_x = \A_{x_1} \cdots \A_{x_t}$ that
will be used throughout the paper. In terms of the graphical description of $A$, the value $f_A(x)$ can be interpreted as the sum of the weights of all paths labeled by $x$ from an initial to a final state, where the weight of a path is the product of the initial weight, the corresponding transition weights, and the final weight:
\begin{equation*}
f_A(x) = \sum_{(q_0,\ldots,q_{t}) \in [n]^{t+1}} \azero(q_0) \left(\prod_{i=1}^t \A_{x_i}(q_{i-1},q_i) \right) \ainf(q_t) \enspace,
\end{equation*}
where $t = |x|$.
A function $f : \sstar \to \R$ is called \emph{rational}\footnote{Some authors call these functions \emph{recognizable} and use a notion of rationality associated with belonging to a set of functions closed under certain operations. Since both notions are equivalent for the computation model of WFA we consider in this paper, we purposefully avoid the distinction between rationality and recognizability.} if there exists a WFA $A$ such that $f = f_A$.
The \emph{rank} of a rational function $f$ is the dimension of the smallest WFA
realizing $f$.
We say that a WFA $A$ is \emph{minimal} if $\dim(A) = \rank(f_A)$.

Hankel matrices provide a powerful characterization of rationality that will be heavily used in the sequel. Let $\H \in \R^{\sstar \times \sstar}$ be an infinite matrix whose rows and
columns are indexed by strings. We say that $\H$ is \emph{Hankel}\footnote{In real analysis a matrix $\mM$ is Hankel if $\mM(i,j)=\mM(k,l)$
whenever $i + j = k + l$, which implies that $\mM$ is symmetric. In our case we
have $\H(p,s) = \H(p',s')$ whenever $p s = p' s'$, but $\H$ is not symmetric
because string concatenation is not commutative whenever $|\Sigma| > 1$.}
if for all strings $p, p', s, s' \in \sstar$ such that $p s = p' s'$ we have
$\H(p,s) = \H(p',s')$.
Given a function $f : \sstar \to \R$ we can associate with it a Hankel matrix
$\H_f \in \R^{\sstar \times \sstar}$ with entries $\H_f(p,s) = f(p s)$.
Conversely, given a matrix $\H \in \R^{\sstar \times \sstar}$ with the Hankel
property, there exists a unique function $f : \sstar \to \R$ such that $\H_f =
\H$.
The following well-known theorem characterizes all Hankel matrices of finite
rank.

\begin{theorem}[\citep{berstel2011noncommutative}]\label{thm:fundamentalWFA}
For any function $f : \sstar \to \R$, the Hankel matrix $\H_f$ has finite rank
$n$ if and only if $f$ is rational with $\rank(f) = n$. In other words,
$\rank(f) = \rank(\H_f)$ for any function $f : \sstar \to \R$.
\end{theorem}

\subsection{Probabilistic Automata}

Probabilistic automata will be used as a recurring example throughout the
paper. Here we introduce the main definitions and stress some key
differences arising from subtle changes in the definition that can make a
difference in terms of the analytic properties of this kind of
automata. Generally speaking, a probabilistic automaton is a WFA $A$ whose
weights have a probabilistic interpretation and such that the values
$f_A(x)$ of the function computed by $A$ represent the likelihood of an
event associated with string $x$.  

A \emph{generative probabilistic automaton} (GPA) is a WFA $A$ such that
the function $f_A$ computes a probability distribution on $\sstar$. That
is, we have $f_A(x) \geq 0$ and $\sum_{x \in \sstar} f_A(x) = 1$. In
addition, we say a GPA $A = \wa$ is \emph{proper} (pGPA) if its weights
have a probabilistic interpretation, i.e.\ 
\begin{enumerate}
\item Initial weights represent a probability distribution over states:
  $\azero \geq 0$ and $\azero^\top \mat{1} = 1$. 
\item Transition weights and final weights represent probabilities of
  emitting a symbol and transitioning to a next state or terminating:
  $\A_\sigma \geq 0$\footnote{Note that these inequalities have scalars in
    their RHS and should be interpreted as entry-wise inequalities, and not
    as claims about positive semi-definite matrices.}, $\ainf \geq 0$, and
  $\sum_{\sigma \in \Sigma} \A_\sigma \mat{1} + \ainf = \mat{1}$. 
\end{enumerate}
An example is provided in Figure~\ref{fig:gpa}.
It is shown in \citep{denis2008rational} that not all GPA are pGPA, and that
there exists probability distributions on $\sstar$ that cannot be computed by any
pGPA. 

A \emph{dynamic probabilistic automaton} (DPA) is a WFA $A$ defining a
probability distribution $D_A$ over streams in $\Sigma^\omega$ and such
that the function $f_A$ on finite strings computes the probability under
$D_A$ of cones of the form $x \Sigma^{\omega}$ for $x \in \sstar$. That is,
we have the semantics $f_A(x) = \P_{D_A}[x \Sigma^{\omega}]$, which implies
that $f_A(\Sigma^t) = 1$ for all $t \geq 0$. Again, we say that a DPA $A =
\wa$ is \emph{proper} (pDPA) if its weights have a probabilistic
interpretation as follows: 
\begin{enumerate}
\item Initial weights represent a probability distribution over states:
  $\azero \geq 0$ and $\azero^\top \mat{1} = 1$. 
\item Final weights are all equal to one: $\ainf = \mat{1}$.
\item Transition weights represent probabilities of emitting a symbol and
  transitioning to a next state: $\A_\sigma \geq 0$ and $\sum_{\sigma \in
    \Sigma} \A_\sigma \mat{1} = \mat{1}$. 
\end{enumerate}
An example is provided in Figure~\ref{fig:dpa}.  As with GPA, there exist
improper DPA, and distributions $D_A$ on $\Sigma^\omega$ that cannot be
computed by any pDPA \citep{denis2008rational}. An important subclass of
pDPA are those for which there is no state with deterministic emissions. A
pDPA $A = \wa$ is \emph{det-free} if we have
$\norm{\A_\sigma}_{\infty} < 1$ for each $\sigma \in \Sigma$. Note that if
$A$ has $n$ states and there exists $\sigma$ such that
$\norm{\A_\sigma}_{\infty} = 1$, then there exists $i \in [n]$ such that
$\A_\sigma \mat{1} (i) = 1$ and therefore from state $i$ the automaton $A$
always emits symbol $\sigma$.

\begin{figure}[t]
\centering
\begin{subfigure}[c]{0.4\textwidth}
\centering
\begin{tikzpicture}[shorten >=1pt,node distance=3cm,auto]
\tikzstyle{every state}=[draw=blue!50,very thick,on grid,fill=blue!20,inner sep=3pt,minimum size=4mm]
\node[state with output,initial,initial text=$1$] (q_0) {$q_1$ \nodepart{lower} $0$};
\node[state with output] (q_1) [right of=q_0] {$q_2$ \nodepart{lower} $1/3$};
\path[->]
(q_0) edge [loop above] node
{\footnotesize $\begin{matrix}{{a , 1/4}}\end{matrix}$} ()
(q_0) edge [bend right] node [swap]
{\footnotesize $\begin{matrix}{{a , 1/4}}\\{b , 1/2}\end{matrix}$} (q_1)
(q_1) edge [loop above] node
{\footnotesize $\begin{matrix}{{a , 1/3}}\end{matrix}$} ()
(q_1) edge [bend right] node [swap]
{\footnotesize $\begin{matrix}{{b , 1/3}}\end{matrix}$} (q_0);
\end{tikzpicture}
\end{subfigure}
\begin{subfigure}[c]{0.4\textwidth}
\centering
\begin{minipage}[t]{\textwidth}
\begin{tabular}{cc}
$\azero =
\left[
\begin{matrix}
1\\
0\\
\end{matrix}
\right]$
&
$\A_a =
\left[
\begin{matrix}
1/4 & 1/4 \\
0 & 1/3\\
\end{matrix}
\right]$\\[.5cm]
$\ainf =
\left[
\begin{matrix}
0\\
1/3\\
\end{matrix}
\right]$
&
$\A_b =
\left[
\begin{matrix}
0 & 1/2 \\
1/3 & 0\\
\end{matrix}
\right]$
\end{tabular}
\end{minipage}
\end{subfigure}
\caption{Example of pGPA $A$ with two states.}
\label{fig:gpa}
\end{figure}
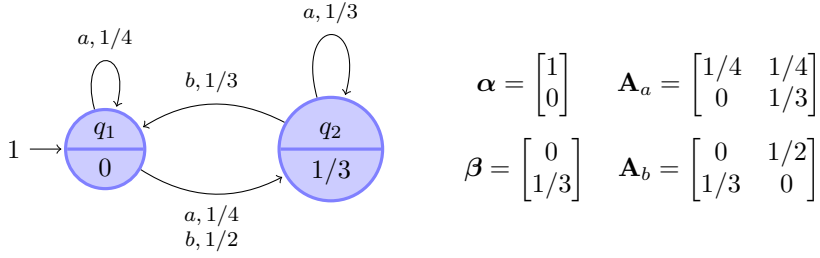

\begin{figure}[t]
\centering
\begin{subfigure}[c]{0.4\textwidth}
\centering
\begin{tikzpicture}[shorten >=1pt,node distance=3cm,auto]
\tikzstyle{every state}=[draw=blue!50,very thick,on grid,fill=blue!20,inner sep=3pt,minimum size=4mm]
\node[state with output,initial,initial text=$1$] (q_0) {$q_1$ \nodepart{lower} $1$};
\node[state with output] (q_1) [right of=q_0] {$q_2$ \nodepart{lower} $1$};
\path[->]
(q_0) edge [loop above] node
{\footnotesize $\begin{matrix}{{a , 1/4}}\end{matrix}$} ()
(q_0) edge [bend right] node [swap]
{\footnotesize $\begin{matrix}{{a , 1/4}}\\{b , 1/2}\end{matrix}$} (q_1)
(q_1) edge [loop above] node
{\footnotesize $\begin{matrix}{{a , 1/3}}\end{matrix}$} ()
(q_1) edge [bend right] node [swap]
{\footnotesize $\begin{matrix}{{b , 2/3}}\end{matrix}$} (q_0);
\end{tikzpicture}
\end{subfigure}
\begin{subfigure}[c]{0.4\textwidth}
\centering
\begin{minipage}[t]{\textwidth}
\begin{tabular}{cc}
$\azero =
\left[
\begin{matrix}
1\\
0\\
\end{matrix}
\right]$
&
$\A_a =
\left[
\begin{matrix}
1/4 & 1/4 \\
0 & 1/3\\
\end{matrix}
\right]$\\[.5cm]
$\ainf =
\left[
\begin{matrix}
1\\
1\\
\end{matrix}
\right]$
&
$\A_b =
\left[
\begin{matrix}
0 & 1/2 \\
2/3 & 0\\
\end{matrix}
\right]$
\end{tabular}
\end{minipage}
\end{subfigure}
\caption{Example of det-free pDPA $A$ with two states.}
\label{fig:dpa}
\end{figure}
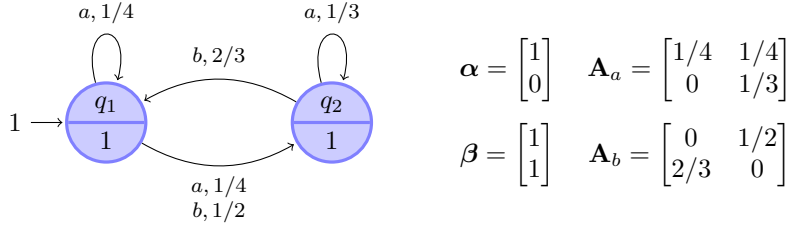

\section{Banach and Hilbert Spaces of Rational Functions}\label{sec:ratfuns}

In the literature on formal language theory, functions $f : \sstar \to \R$
are sometimes regarded as weighted languages and weighted automata
computing them as linear representations.  From an algebraic point of view,
one can identify a weighted language $f$ with an element of
the vector space $\R^{\sstar}$. 
This vector space contains several subspaces that play an important role in
the theory developed in this paper.  Furthermore, some of these spaces can
be endowed with additional operations and norms, yielding a wide variety of
algebraic and analytic structures.  To the best of our knowledge, analytic
properties of these spaces have never been systematically studied before in
the automata theory literature.  This section introduces the basic facts
and definitions that will be needed in the rest of the paper.  We also take
this as an opportunity to prove basic facts about these spaces and pinpoint
ideas that need to be developed further.  Overall, we hope this provides
the foundations for a much needed \emph{analytic theory of rational
  functions}.

A fundamental linear subspace of $\R^{\sstar}$ is the space of all rational
functions, which we denote by $\cR(\Sigma)$. 
That $\cR(\Sigma)$ is a linear subspace follows from the simple observations that if $f, g
\in \cR(\Sigma)$ and $c \in \R$, then $c f$ and $f + g$ are both rational
\citep{berstel2011noncommutative}. 
An important subspace of $\cR(\Sigma)$ is the space of all $f \in
\R^{\sstar}$ with finite support, which we denote by
$\finsup(\Sigma)$. That is, $f \in \finsup(\Sigma)$ if and only if
$|\supp(f)| < \infty$, where $\supp(f) = \{ x : f(x) \neq 0 \}$ is the
support of $f$. It is immediate from this definition that $\finsup(\Sigma)$
is a linear subspace of $\R^{\sstar}$. The containment $\finsup(\Sigma)
\subset \cR(\Sigma)$ follows from observing that every function with finite
support is rational \citep{berstel2011noncommutative}. 

Another important family of subspaces of $\R^{\sstar}$ are the ones containing
all functions with finite $p$-norm for some $1 \leq p \leq \infty$, which is
given by $\norm{f}_p^p = \sum_{x \in \sstar} |f(x)|^p$ for finite $p$, and
$\norm{f}_\infty = \sup_{x \in \sstar} |f(x)|$; we denote this space by
$\lp(\Sigma)$. 
Note that these are Banach spaces, and as with the usual theory of Banach
spaces over sequences we have 
$\lp(\Sigma) \subset \lpq(\Sigma)$ for $p < q$.
Of these, $\ltwo(\Sigma)$ can be endowed with the structure of a separable
Hilbert space with the inner product $\langle f, g \rangle = \sum_{x \in \sstar}
f(x) g(x)$.  Recall that in this case we have the \emph{Cauchy--Schwarz
inequality} $\langle f, g \rangle^2 \leq \norm{f}_2^2 \, \norm{g}_2^2$.
In addition, we have its generalization, \emph{H{\"o}lder's inequality}: given $f
\in \lp(\Sigma)$ and $g \in \lpq(\Sigma)$ with $p^{-1} + q^{-1} \leq 1$, then
$\norm{f \cdot g}_1 \leq \norm{f}_p \norm{g}_q$, where $(f \cdot g)(x) =
f(x) g(x)$ is the \emph{Hadamard product} between two languages.

By intersecting any of the previous subspaces with $\cR(\Sigma)$ one obtains
$\lprat(\Sigma) = \cR(\Sigma) \cap \lp(\Sigma)$, the normed vector space containing
all rational functions with finite $p$-norm.
In most cases the alphabet $\Sigma$ will be clear from the context and we will
just write $\cR$, $\finsup$, $\lp$, and $\lprat$.
It is important to note that although the $\lp$ spaces can be endowed with
the structure of a Banach or Hilbert space, the $\lprat$ spaces cannot, because
they are not complete; i.e.\ it is possible to find sequences of functions
in $\lprat$ whose limit in the topology induced by the corresponding norm
is not rational. 
For example, consider the function given by $f(x) = (k+1)^{-|x|}$ if $x$ is
a palindrome and $f(x) = 0$ otherwise. Since $\supp(f)$ is the set of all
palindromes then $f$ is not rational \citep{berstel2011noncommutative}, and
in addition $\norm{f}_1 < \infty$ by construction. Thus, we have we have $f
\in \lp \setminus \cR$ for any $1 \leq p \leq \infty$. Now, for any $l \geq
0$ let $f_l(x) = f(x)$ if $|x| \leq l$ and $f_l(x) = 0$ otherwise. Since
$f_l$ has finite support for every $l \geq 0$ we have $f_l \in
\lprat$. Finally, it is easy to check that $\lim_{l \to \infty} \norm{f -
  f_l}_p = 0$, implying that we have a sequence of functions in $\lprat$
converging to a non-rational function. Therefore none of the $\lprat$
spaces is complete.  Nonetheless, the following result shows that all $\lp$
spaces with $1 \leq p < \infty$ can be obtained as the completion of their
corresponding $\lprat$ space. 

\begin{theorem}
For any $1 \leq p < \infty$, the Banach space $\lp$ coincides with the
completion of $\lprat$ with respect to $\norm{\cdot}_p$. 
\end{theorem}
\begin{proof}
Fix $1 \leq p < \infty$. Since $\finsup \subset \lprat$, it suffices to
show that $\finsup$ is dense in $\lp$ with respect to the topology induced
by $\norm{\cdot}_p$. Let $f \in \lp$ and for any $l \geq 0$ define $f_l(x)
= f(x)$ if $|x| \leq l$ and $f_l(x) = 0$ otherwise. Clearly we have $f_l
\in \finsup$ by construction. To see that $f_l \to f$ in the topology of
$\lp$ as $l \to \infty$ we write $s_l = \norm{f_l - f}_p^p = \sum_{|x| > l}
|f(x)|^p$ and observe that we must have $s_l \to 0$. Otherwise we would
have $\lim_{l \to \infty} \sum_{|x| = l} |f(x)|^p > 0$, which is a
contradiction with $\norm{f}_p^p = \sum_{x \in \sstar} |f(x)|^p < \infty$. 
\end{proof}

\subsection{Bounded Hankel Operators}

Recall that Theorem~\ref{thm:fundamentalWFA} gives a characterization of
the functions $f : \sstar \to \R$ which have a Hankel matrix with finite
rank. Using the concepts introduced above we can interpret the Hankel
matrix as an operator on Hilbert spaces and ask when this operator
satisfies some nice properties. The main result of this section is a
characterization of the rational functions whose Hankel matrix induces a
bounded operator on $\ltwo$. 

Recall that a matrix $\mT \in \R^{\sstar \times \sstar}$ can be interpreted as the expression of a (possibly unbounded) linear operator $T : \ltwo \to \ltwo$ in terms of the canonical basis $(\me_x)_{x \in \sstar}$.
In the case of a Hankel matrix $\H_f$, we can see it is associated with an operator $H_f$ corresponding to the operation $g \mapsto H_f g$ with $(H_f g)(x) = \sum_y f(x y) g(y)$ (assuming the series converges).
An operator $T : \ltwo \to \ltwo$ is bounded if $\normop{T} = \sup_{\norm{g}_2 \leq 1} \norm{T g}_2 < \infty$.
Not all Hankel operators $H_f$ are bounded, but we shall give a necessary and sufficient condition for $H_f$ to be bounded when $f$ is rational. We start with the following a technical lemma.

\begin{lemma}\label{lem:denis}
Let $A = \wa$ be a WFA such that $f_A(x) \geq 0$ for all $x \in \sstar$. Define $\A = \sum_\sigma \A_\sigma$ and let $\rho = \rho(\A)$ be its spectral radius. Then the following hold:
\begin{enumerate}
\item If $A$ is minimal and $f_A \in \lonerat$, then $\rho < 1$.
\item If $\rho < 1$, then $f_A \in \lonerat$.
\end{enumerate}
\end{lemma}
\begin{proof}
We start by recalling that if $A = \wa$ is a minimal WFA with $n$ states,
then there exist sets of prefixes $\Ps = \{p_1,\ldots,p_n\}$ and suffixes
$\Ss = \{s_1,\ldots,s_n\}$ such that the sets of vectors $\{\azero^\top
\A_{p_1}, \ldots, \azero^\top \A_{p_n}\}$ and $\{\A_{s_1} \ainf, \ldots,
\A_{s_n} \ainf\}$ define two bases for $\R^n$
\citep{berstel2011noncommutative}. 
For convenience we will write $\balpha_{p_i}^\top = \azero^\top \A_{p_i}$
and $\bbeta_{s_j} = \A_{s_j} \ainf$. 

Now assume $f_A \in \lonerat$ and suppose $\lambda$ is an arbitrary
eigenvalue of $\A$. We need to show that $|\lambda| < 1$. Let $\v$ be any
eigenvector with eigenvalue $\lambda$ and suppose $\norm{\v}_2 = 1$. Using
the basis given by $\Ps$ and $\Ss$ we can find coefficients such that $\v =
\sum_{i \in [n]} \gamma_i \balpha_{p_i} = \sum_{j \in [n]} \delta_j
\bbeta_{s_j}$. For any $k \geq 0$ we can now write the following: 
\begin{align*}
|\lambda|^k &= |\lambda|^k \v^\top \v = \left|\v^\top (\lambda^k \v)\right| =
\left|\v^\top \A^k \v\right| \\
&=
\left|\left(\sum_i  \gamma_i \balpha_{p_i}^\top \right) \left(\sum_\sigma \A_\sigma\right)^k \left(\sum_j \delta_j \bbeta_{s_j}\right)\right| \\
&\leq
\sum_{i,j} |\gamma_i| |\delta_j| \sum_{x \in p_i \Sigma^k s_j} |f_A(x)| \enspace.
\end{align*}
Since we have $f_A \in \lonerat$ by hypothesis, for fixed $i$ and $j$ we have $\sum_{k \geq 0} \sum_{x \in p_i \Sigma^k s_j} |f_A(x)| \leq \sum_{x \in \sstar} |f_A(x)| < \infty$. Therefore we can conclude that $\sum_{k \geq 0} |\lambda|^k < \infty$, which necessarily implies $|\lambda| < 1$.

To obtain the converse suppose $\rho(\A) < 1$ and note that because $f_A$ is non-negative we have
\begin{equation}
\norm{f_A}_1 = \sum_{x \in \sstar} |f(x)| = \sum_{x \in \sstar} f(x) = \sum_{k \geq 0} \azero^\top \A^k \ainf < \infty \enspace.
\end{equation}
Note that this implication does not require the minimality of $A$.
\end{proof}

The following theorem is the main result of this section.

\begin{theorem}\label{thm:iffsvd}
Let $f : \sstar \to \R$ be a rational function. The Hankel operator $H_f$ is bounded if and only if $f \in \ltworat$.
\end{theorem}
\begin{proof}
It is easy to see that the membership $f \in \ltwo$ is a necessary condition for the boundedness of $\H_f$. Indeed, by noting that $f$ appears as the first column of $\H_f$ we have $f = H_f e_{\varepsilon}$, and since $\norm{e_{\varepsilon}}_2 = 1$ we have $\norm{f}_2 = \norm{H_f e_{\varepsilon}}_2 \leq \normop{H_f}$.

Next we prove sufficiency.
Let $g \in \ltwo$ with $\norm{g}_2 = 1$ and for any $x \in \sstar$ define the function $f_x(y) = f(x y)$. With this notation we can write
\begin{align}
\norm{H_f g}_2^2 &= \sum_{x \in \sstar} \left(\sum_{y \in \sstar} f(x y) g(y)\right)^2
= \sum_{x \in \sstar} \langle f_x, g \rangle^2 \notag \\
&\leq \norm{g}_2^2 \sum_{x \in \sstar} \norm{f_x}_2^2
= \sum_{x \in \sstar} \sum_{y \in \sstar} f(x y)^2 \notag \\
&= \sum_{z \in \sstar} (1+|z|) f(z)^2 \label{eqn:bounded}
\enspace,
\end{align}
where we used Cauchy--Schwarz's inequality, and the fact that a string $z$ can be written as $z = xy$ in $1 + |z|$ different ways.

Recall that $f \in \ltworat$ implies $f^2 \in \lonerat$. Let $A = \wa$ be a
minimal WFA for $f^2$ and write $\A = \sum_{\sigma} \A_\sigma$. Note we
have $\rho = \rho(\A) < 1$ by Lemma~\ref{lem:denis}. Suppose $\A = \mW \mJ
\mW^{-1}$ is the Jordan canonical form of $\A$ and let $m$ denote the
maximum algebraic multiplicity of any eigenvalue of $\A$. By computing the $k$th power of the largest Jordan block
\begin{align*}
\mJ_{\lambda} =
\left[
\begin{array}{ccccc}
\lambda & 1       & 0       & \cdots  & 0 \\
0       & \lambda & 1       & \cdots  & 0 \\
\vdots  & \vdots  & \vdots& \ddots  & \vdots \\
0       & 0       & 0        & \lambda & 1       \\
0       & 0       & 0       & 0       & \lambda
\end{array}
\right] \in \R^{m' \times m'}
\end{align*}
associated with the maximal eigenvalue $|\lambda| = \rho$ (with $m' \leq m$)
one can see there exists
a constant $c > 0$ such that the following holds for all $k \geq 0$: 
\begin{equation*}
\sum_{x \in \Sigma^k} f(x)^2 = \azero^\top \A^k \ainf =
\azero^\top \mW \mJ^k \mW^{-1} \ainf \leq c k^{m-1} \rho^k \enspace.
\end{equation*}
This is a standard calculation in the analysis of non-reversible Markov chains; see Fact 3 in \citep{rosenthal1995convergence} for more details.
Now we use that $\rho < 1$, in which case this bound yields
\begin{equation*}
\sum_{z \in \sstar} |z| f(z)^2 = \sum_{k \geq 0} k \sum_{z \in \Sigma^k} f(z)^2 \leq c \sum_{k \geq 0} k^m \rho^k < \infty \enspace.
\end{equation*}
Plugging this into \eqref{eqn:bounded} we can conclude that $\norm{H_f g}_2$ is finite and therefore $H_f$ is bounded.
\end{proof}

\section{The Singular Value Automaton}\label{sec:sva}

The central object of study in this paper is the singular value automaton
(SVA). Essentially, this is a canonical form for weighted automata which is
tightly connected to the singular value decomposition of the corresponding
Hankel matrix. We will start this section by establishing some fundamental
preliminary results on the relation between minimal WFAs and rank
factorizations of Hankel matrices. By assuming that one such Hankel matrix
admits a singular value decomposition, the relation above will lead us
directly to the definition of singular value automaton. We then proceed to
explore necessary conditions for the existence of SVA. These will
essentially say that only rational functions in $\ltworat$ admit a singular
value automaton, provide some easily testable conditions, and guarantee the
existence of an SVA for a large class of probabilistic automata.

\subsection{Correspondence between Minimal WFA and Rank Factorizations}

An important operation on WFA is conjugation by an invertible matrix.
Let $A = \wa$ be a WFA of dimension $n$ and suppose $\mQ \in \R^{n \times n}$ is
invertible. Then we can define the \emph{conjugate} of $A$ by $\mQ$ as:
\begin{equation}\label{eqn:wfaconjugation}
A' = A^{\mQ} = \mQ^{-1} A \mQ = \waQ \enspace.
\end{equation}
It follows immediately that $f_A = f_{A'}$ since, at every step in the
computation of $f_{A'}(x)$, the products $\mQ \mQ^{-1}$ vanish. 
This means that the function computed by a WFA is invariant under conjugation,
and that given a rational function $f$, there exist infinitely many WFA realizing $f$.
The following result offers a full characterization of all minimal WFA
realizing a particular rational function. 

\begin{theorem}[\citep{berstel2011noncommutative}]\label{thm:conjugacy}
If $A$ and $B$ are minimal WFA realizing the same function, then $B =
A^{\mQ}$ for some invertible $\mQ$. 
\end{theorem}

The goal of this section is to provide a ``lifted'' version of this result
establishing a connection between every pair of rank factorizations of the
Hankel matrix $\H_f$, and then show that these rank factorizations are in
bijection with all minimal WFA for $f$.  We start by recalling how every
minimal WFA realizing $f$ induces a rank factorization for $\H_f$.

Suppose $f$ is a rational function and $A = \wa$ is a WFA realizing $f$.
The \emph{forward matrix} of $A$ is defined as the infinite matrix $\mP_A
\in \R^{\sstar \times n}$ with entries given by $\mP_A(p,:) = \azero^\top
\A_p$ for any string $p \in \sstar$; sometimes we will refer to the strings
indexing rows in a forward matrix as \emph{prefixes}. 
Similarly, let $\mS_A \in \R^{\sstar \times n}$ be the \emph{backward
  matrix} of $A$ given by $\mS_A(s,:) = (\A_s \ainf)^\top$ for any string
$s \in \sstar$; strings indexing rows in a backward matrix are commonly
called \emph{suffixes}. 
Now note that for every $p, s \in \sstar$ we have
\begin{equation}
\H_f(p,s) = f(p s) = (\azero^\top \A_p) \left(\A_s \ainf \right) = \sum_{i
  \in [n]} \mP_A(p,i) \mS_A(s,i) = \mP_A(p,:) \mS_A^\top(:,s) \enspace. 
\end{equation}
Therefore, we see that the forward and backward matrix of $A$ yield the
factorization $\H_f = \mP_A \mS_A^\top$. 
This is known as the \emph{forward--backward} (FB) factorization of $\H_f$
induced by $A$ \citep{mlj13spectral}. 

Recall that a WFA $A$ with $n$ states is called \emph{reachable} when the
space spanned by all the forward vectors has dimension $n$; that is: 
\begin{equation}
\dim \linspan \{ \azero^\top \A_x \;|\; x \in \sstar \} = \rank(\mP_A) = n \enspace.
\end{equation}
Similarly, $A$ is called \emph{observable} if the dimension of the space
spanned by the backward vectors equals $n$; that is: 
\begin{equation}
\dim \linspan \{ \A_x \ainf \;|\; x \in \sstar \} = \rank(\mS_A) = n \enspace.
\end{equation}
Note that when $A$ is minimal, the number of columns of the forward and
backward matrices equals the rank of $\H_f$, and therefore the FB
factorization is a rank factorization. Therefore, it follows from
Theorem~\ref{thm:fundamentalWFA} the useful characterization of minimality
saying that a WFA $A$ is minimal if and only if it is both reachable and
observable. 

The following result shows that every rank factorization of $\H_f$ is
actually an FB factorization. We can understand this result as a refinement
of Theorem~\ref{thm:fundamentalWFA} in the sense that given a finite-rank
Hankel matrix, it provides a characterization of all its possible rank
factorizations. 

\begin{proposition}\label{prop:duality}
Let $f$ be rational and suppose $\H_f = \mP \mS^\top$ is a rank factorization. Then
there exists a minimal WFA $A$ realizing $f$ which induces this factorization.
\end{proposition}
\begin{proof}
Let $B$ be any minimal WFA realizing $f$ and denote $n = \rank(f)$. Then we have
two rank factorizations $\mP \mS^\top = \mP_B \mS_B^\top$ for the Hankel matrix
$\H_f$.
Therefore, the columns of $\mP$ and $\mP_B$ both span the same $n$-dimensional
sub-space of $\R^{\sstar}$, and there exists a change of basis $\mQ \in \R^{n
\times n}$ such that $\mP_B \mQ = \mP$. This implies we must also have
$\mS^\top = \mQ^{-1} \mS_B^\top$.
It follows that $A = B^{\mQ}$ is a minimal WFA for $f$ inducing the desired
rank factorization.
\end{proof}

\subsection{Definition of Singular Value Automaton}

It is well-known that the compact singular value decomposition of a matrix
is a rank-revealing decomposition in the sense that the intermediate
dimensions of the decomposition correspond to the rank of the matrix. This
decomposition can be used to construct rank factorizations for said
matrix. The singular value automaton links this idea with the minimal WFA
identified in Proposition~\ref{prop:duality}.

Recall that if $\H_f$ is a Hankel matrix of rank $n$ admitting a singular
value decomposition, then there exists a square matrix
$\mD = \diag(\sval_1, \ldots, \sval_n) \in \R^{n \times n}$ and two
infinite matrices $\mU, \mV \in \R^{\sstar \times n}$ with orthonormal
columns (i.e.\ $\mU^\top \mU = \mV^\top \mV = \mI$) such that
$\H_f = \mU \mD \mV^\top$ with $\mU, \mV \in \R^{\sstar \times n}$. By
splitting this decomposition into two parts we obtain the rank
factorization $\H_f = (\mU \mD^{1/2}) (\mV \mD^{1/2})^\top$. Thus, whenever
$\H_f$ admits an SVD, we can invoke Proposition~\ref{prop:duality} to
conclude there exists a minimal WFA realizing $f$ whose induced FB rank
factorization coincides with the one we obtained above from SVD. Putting
this into a formal statement we get the following theorem.

\begin{theorem}\label{thm:sva}
  Let $f$ be a rational function and suppose $\H_f$ admits a compact SVD
  $\H_f = \mU \mD \mV^\top$. Then there exists a minimal WFA $A$ for $f$
  inducing the rank factorization
  $\H_f = (\mU \mD^{1/2}) (\mV \mD^{1/2})^\top$. That is, $A$ is a WFA for
  $f$ with FB rank factorization given by $\mP_A = \mU \mD^{1/2}$ and
  $\mS_A = \mV \mD^{1/2}$.
\end{theorem}

The WFA given by the above theorem can be considered as a canonical form
for a rational function whose Hankel matrix admits an SVD. This is made
formal in the following definition. Next section will provide conditions
for the existence of this automaton. 

\begin{definition}
Let $f \in \ltworat$. A \emph{singular value automaton} (SVA) for $f$ is a
minimal WFA $A$ realizing $f$ such that the FB rank factorization of $\H_f$
induced by $A$ has the form given in Theorem~\ref{thm:sva}.
\end{definition}

Note the SVA provided by Theorem~\ref{thm:sva} is unique up to the same
conditions in which SVD is unique. In particular, it is easy to verify that
if the Hankel singular values of $f \in \ltworat$ satisfy the strict
inequalities $\sval_1 > \cdots > \sval_n$, then the transition weights of
the SVA $A$ of $f$ are uniquely defined, and the initial and final weights
are uniquely defined up to sign changes. 

\subsection{Rational Functions Admitting an SVA}\label{sec:ratfunswithsva}

By leveraging the fact that every compact operator on a Hilbert space
admits a singular value decomposition and our Theorem~\ref{thm:iffsvd}
characterizing rational functions with bounded Hankel operator, we
immediately get a characterization of rational functions admitting an SVA. 

\begin{theorem}
A rational function $f : \sstar \to \R$ admits an SVA if and only if $f \in \ltworat$.
\end{theorem}
\begin{proof}
Since a finite-rank bounded operator is compact and therefore admits a compact SVD, Theorem~\ref{thm:iffsvd} and Theorem~\ref{thm:sva} imply that every $f \in \ltworat$ admit an SVA. On the other hand, if a rational function admits an SVA, then its Hankel $\H_f$ matrix admits a compact SVD and therefore $H_f$ is bounded. Applying Theorem~\ref{thm:iffsvd} we see that this implies $f \in \ltworat$.
\end{proof}

In view of this result, when given a rational function as a WFA, one just has to check that the function has finite $\ltwo$ norm to ensure the existence of an SVA for that function. A direct way to test this based on Lemma~\ref{lem:denis} is given below.

\begin{theorem}\label{thm:l2minimal}
Let $A$ be a WFA and let $B$ be a minimisation of the automaton $A \otimes A$ computing $f_A^2$. Then we have $f_A \in \ltworat$ if and only if $\rho(\sum_{\sigma \in \Sigma} \B_\sigma) < 1$.
\end{theorem}
\begin{proof}
Let $\B = \sum_{\sigma \in \Sigma} \B_\sigma$. The if part follows from observing that $\rho(\sum_{\sigma \in \Sigma} \B_\sigma) < 1$ implies that $\sum_{x \in \sstar} \B_x = \sum_{t \geq 0} \B^t$ converges, and therefore $\norm{f_A}_2^2 = \sum_{x \in \sstar} \bzero^\top \B_x \binf$ is finite. The only if part is a direct application of Lemma~\ref{lem:denis}.
\end{proof}

The above theorem gives a direct way to check if for a given $A$ we have $f_A \in \ltworat$ by using a WFA minimisation algorithm and computing the spectral radius of a given matrix. If $A$ has $n$ states, then $B$ can be obtained by minimising an automaton with $n^2$ states, which takes time $O(n^6)$ \citep{berstel2011noncommutative} and yields a WFA $B$ with $n' \leq n^2$ states. Computing the spectral radius of $\B$ takes time $O(n'^3)$ \citep{trefethen1997numerical}, so the overall complexity of testing $f_A \in \ltworat$ based in the above theorem is $O(n^6)$. The following theorem gives sufficient conditions for $f_A \in \ltworat$, some of which can be checked without the need to run a WFA minimisation algorithm.

\begin{theorem}\label{thm:l2sufficient}
Let $A$ be a WFA computing a function $f_A$. Any of the following conditions implies $f_A \in \ltworat$:
\begin{enumerate}
\item \label{it:0} $f_A \in \lonerat$,
\item \label{it:1} $\rho(\sum_{\sigma} \A_\sigma \otimes \A_\sigma) < 1$,
\item \label{it:2} $\norm{\sum_{\sigma} \A_\sigma \otimes \A_\sigma}_p < 1$ for some $1 \leq p \leq \infty$,
\item \label{it:3} $\norm{\sum_{\sigma} \A_\sigma \A_\sigma^\top}_2 < 1$.
\end{enumerate}
\end{theorem}
\begin{proof}
The first item follows from the inclusion $\lonerat \subset \ltworat$. To get \eqref{it:1} note that by Lemma~\ref{lem:denis} the condition implies $f_A^2 \in \lonerat$ and therefore $f_A \in \ltworat$. Condition \eqref{it:2} follows from the property of the spectral radius $\rho(\mM) \leq \norm{\mM}_p$. The last condition follows from the main result in \citep{lototsky2015simple} showing that $\rho(\sum_{\sigma} \A_\sigma \otimes \A_{\sigma}) \leq \norm{\sum_{\sigma} \A_{\sigma} \A_{\sigma}^\top}_2$.
\end{proof}

We can use these conditions to identify classes of probabilistic automata that compute functions in $\ltworat$ and therefore have an SVA. We will need the following technical lemma.

\begin{lemma}\label{lem:norminfkron}
The following inequality holds for any set of square matrices $\{\A_1, \ldots, \A_m\}$:
\begin{align*}
\bnorm{\sum_{k \in [m]} \A_k \otimes \A_k}_\infty &\leq \bnorm{[\A_{1} \ldots \A_{m}]}_\infty \left\| \left[\begin{array}{c} \A_{1} \\ \vdots \\ \A_{m} \end{array} \right] \right\|_{\infty} \\
&= \bnorm{[\A_{1} \ldots \A_{m}]}_\infty \bnorm{[\A_{1}^\top \ldots \A_{m}^\top]}_1
\enspace.
\end{align*}
\end{lemma}
\begin{proof}
Recall that the induced matrix norm with $p = \infty$ is given by $\norm{\mM}_\infty = \max_i \sum_j |\mM(i,j)|$. Then the desired inequality can be obtained as follows:
\begin{align*}
\bnorm{\sum_{k \in [m]} \A_k \otimes \A_k}_\infty &=
\max_{i_1,i_2 \in [n]} \sum_{j_1,j_2 = 1}^n \left| \sum_{k} \A_k(i_1,j_1) \A_k(i_2,j_2) \right| \\
&\leq \max_{i_1,i_2 \in [n]} \sum_k \sum_{j_1,j_2 = 1}^n |\A_k(i_1,j_1)| |\A_k(i_2,j_2)| \\
&= \max_{i_1,i_2 \in [n]} \sum_k \left(\sum_{j_1= 1}^n |\A_k(i_1,j_1)|\right)\left(\sum_{j_2 = 1}^n |\A_k(i_2,j_2)|\right) \\
&\leq \max_{i_1} \sum_k \left(\sum_{j_1= 1}^n |\A_k(i_1,j_1)|\right)\left(\max_{i_2} \sum_{j_2 = 1}^n |\A_k(i_2,j_2)|\right) \\
&= \max_{i} \sum_k \norm{\A_k}_\infty \left(\sum_{j= 1}^n |\A_k(i,j)|\right) \\
&\leq \left(\max_k \norm{\A_k}_{\infty} \right) \left(\max_{i} \sum_k \sum_{j= 1}^n |\A_k(i,j)|\right) \\
&= \left\| \left[\begin{array}{c} \A_{1} \\ \vdots \\ \A_{m} \end{array} \right] \right\|_{\infty} \norm{[\A_{1} \ldots \A_{m}]}_\infty \enspace.
\end{align*}
The second equality follows from the duality between the norms $\norm{\cdot}_1$ and $\norm{\cdot}_\infty$.
\end{proof}

\begin{corollary}\label{cor:probinltwo}
If $A$ is a GPA or a det-free pDPA, then $f_A \in \ltworat$.
\end{corollary}
\begin{proof}
For $A$ GPA it follows directly from Theorem~\ref{thm:l2sufficient} by noting that we have $\norm{f_A}_1 = 1$. Now suppose $A = \wa$ be a det-free pDPA, so by construction we have $\sum_{a \in \Sigma} \A_a \mat{1} = \mat{1}$ and $\norm{\A_a}_{\infty} < 1$ for all $a \in \Sigma$. Note that the first property implies $\norm{[\A_{a_1} \ldots \A_{a_k}]}_\infty = 1$ and the second property implies
\begin{equation}
\left\| \left[\begin{array}{c} \A_{a_1} \\ \vdots \\ \A_{a_k} \end{array} \right] \right\|_{\infty} < 1 \enspace.
\end{equation}
Therefore, using Lemma~\ref{lem:norminfkron} we see that $\norm{\sum_{a \in \Sigma} \A_a \otimes \A_a}_\infty < 1$ and therefore by \eqref{it:2} in Theorem~\ref{thm:l2sufficient} we get $f_A \in \ltworat$.
\end{proof}

Note that the det-free condition on pDPA is necessary to ensure $f_A \in \ltworat$ as witnessed by the example in Figure~\ref{fig:notl2dpa}.

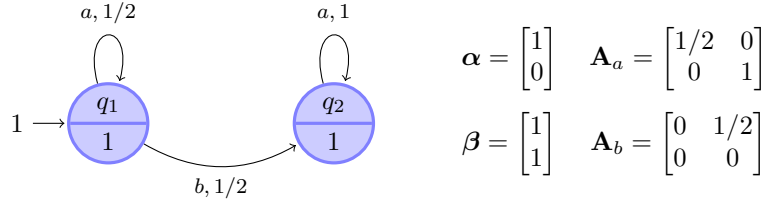
\begin{figure}[t]
\centering
\begin{subfigure}[c]{0.4\textwidth}
\centering
\begin{tikzpicture}[shorten >=1pt,node distance=3cm,auto]
\tikzstyle{every state}=[draw=blue!50,very thick,on grid,fill=blue!20,inner sep=3pt,minimum size=4mm]
\node[state with output,initial,initial text=$1$] (q_0) {$q_1$ \nodepart{lower} $1$};
\node[state with output] (q_1) [right of=q_0] {$q_2$ \nodepart{lower} $1$};
\path[->]
(q_0) edge [loop above] node
{\footnotesize $\begin{matrix}{a , 1/2}\end{matrix}$} ()
(q_0) edge [bend right] node [swap]
{\footnotesize $\begin{matrix}{b , 1/2}\end{matrix}$} (q_1)
(q_1) edge [loop above] node
{\footnotesize $\begin{matrix}{a , 1}\end{matrix}$} ();
\end{tikzpicture}
\end{subfigure}
\begin{subfigure}[c]{0.4\textwidth}
\centering
\begin{minipage}[t]{\textwidth}
\begin{tabular}{cc}
$\azero =
\left[
\begin{matrix}
1\\
0\\
\end{matrix}
\right]$
&
$\A_a =
\left[
\begin{matrix}
1/2 & 0 \\
0 & 1\\
\end{matrix}
\right]$\\[.5cm]
$\ainf =
\left[
\begin{matrix}
1\\
1\\
\end{matrix}
\right]$
&
$\A_b =
\left[
\begin{matrix}
0 & 1/2 \\
0 & 0\\
\end{matrix}
\right]$
\end{tabular}
\end{minipage}
\end{subfigure}
\caption{Example of pDPA $A$ with two states which is not det-free. Note that $f_A(ba^k) = 1/2$ for all $k \geq 0$ and therefore $f_A \notin \ltworat$.}
\label{fig:notl2dpa}
\end{figure}

\section{Fundamental Equations of SVA}\label{sec:fundeqns}

In this section we establish two fundamental facts about SVA that follow from a systematic study of the properties of its observability and reachability Gramian matrices (cf.\ definitions in Section~\ref{sec:gramians}). These matrices, which can be defined for any WFA realizing a function in $\ltworat$, bear a strong relation with the change of basis needed to transform an arbitrary minimal WFA into its SVA form. By studying this relation we will derive an efficient algorithm for the computation of SVA canonical forms provided that we know how to compute the Gramians associated with a WFA. Two algorithms for computing such Gramians are developed in Section~\ref{sec:computegramians}. The second of these algorithms is based on fixed-point equations for the Gramians that are derived in Section~\ref{sec:fpeqns}, which also play a key role on the analysis of an approximate minimisation approach given in Section~\ref{sec:approxmin}.

\subsection{Observability and Reachability Gramians}\label{sec:gramians}

Let $f$ be rational function and $\H_f = \mP \mS^\top$ be a FB factorization for the Hankel matrix of $f$ induced by a (non-necessarily minimal) WFA $A$ with $n$ states. Suppose that $\mP$ is such that the inner products of its columns $\left< \mP(:,i), \mP(:,j) \right> = \sum_{x \in \sstar} \mP(x,i) \mP(x,j)$ are finite for every $i, j \in [n]$. Then the positive semidefinite matrix $\mG_p = \mP^\top \mP \in \R^{n \times n}$ is well-defined. We call $\mG_p$ the \emph{reachability gramian} of $A$. Similarly, suppose the same condition on the inner products holds for the columns of $\mS$. Then the matrix $\mG_s = \mS^\top \mS \in \R^{n \times n}$ is well-defined and we will call it the \emph{observability gramian} of $A$. These definitions are motivated by the following result.

\begin{proposition}\label{prop:rankgramians}
Let $A$ be a WFA with $n$ states and suppose that its reachability and observability gramians are well-defined. Then the following hold:
\begin{enumerate}
\item $A$ is reachable if and only if $\rank(\mG_p) = n$;
\item $A$ is observable if and only if $\rank(\mG_s) = n$;
\item $A$ is minimal if and only if $\rank(\mG_p) = \rank(\mG_s) = n$.
\end{enumerate}
\end{proposition}
\begin{proof}
Recall that $A$ is reachable whenever $\rank(\mP) = n$, which implies that $\mG_p$ is the gramian of $n$ linearly independent vectors and therefore $\rank(\mG_p) = n$. On the other hand, if $\rank(\mG_p) = n$, then by the bound on the rank of a product of matrices we have
\begin{equation}
n = \rank(\mG_p) = \rank(\mP^\top \mP) \leq \max\{ \rank(\mP^\top), \rank(\mP) \} = \rank(\mP) \leq n \enspace,
\end{equation}
from where we conclude that $\rank(\mP) = n$ and therefore $A$ is reachable.

The observable case follows exactly the same reasoning, and the claim about minimality is just a consequence of recalling that $A$ is minimal if and only if it is both reachable and observable.
\end{proof}

Note the above result assumed the gramians are well-defined in the first place. Nonetheless, a similar result can be obtained without such assumptions if one is willing to work with finite versions of these matrices obtained by summing only strings up to some fixed (large enough) length. In particular, defining for any $t \geq 0$ the matrices
\begin{align}
\mG_p^{(t)} &= \sum_{x \in \Sigma^{\leq t}} \mP(x,:)^\top \mP(x,:) \label{eq:gpt} \enspace, \\
\mG_s^{(t)} &= \sum_{x \in \Sigma^{\leq t}} \mS(x,:)^\top \mS(x,:) \label{eq:gst} \enspace,
\end{align}
it is possible to see that when $t \geq n$ we have $\rank(\mG_p^{(t)}) = \rank(\mP)$ and $\rank(\mG_s^{(t)}) = \rank(\mS)$. However, we shall not pursue this direction here. Instead we look for necessary and sufficient conditions guaranteeing the finitness of the gramian matrices.

\begin{proposition}\label{prop:definedgramians}
Let $A$ be a minimal WFA realizing a rational function $f$. The reachability and observability gramians of $A$ are well-defined if and only if $f \in \ltworat$.
\end{proposition}
\begin{proof}
Suppose $A$ is a minimal WFA with $n$ states realizing a function $f \in \ltworat$. It follows from Theorems~\ref{thm:sva} and~\ref{thm:iffsvd} that there exists an invertible matrix $\mQ$ such that $B = A^{\mQ}$ is an SVA. Since $B$ induces the FB factorization given by $\mP_B = \mU \mD^{1/2}$ and $\mS_B = \mV \mD^{1/2}$, we see that the corresponding gramian matrices are well-defined and since $\mU^\top \mU = \mV^\top \mV = \mI$ we have $\mG_{B,p} = \mG_{B,s} = \mD$. Now recall that the FB factorization induced by $A$ has $\mP_A \mQ = \mP_B$ and $\mQ^{-1} \mS_A^\top = \mS_B^\top$. Therefore the gramian matrices associated with $A$ are also well-defined since they can be obtained as $\mG_{A,p} = \mQ^{-\top} \mG_{B,p} \mQ^{-1}$ and $\mG_{A,s} = \mQ \mG_{B,s} \mQ^\top$.

Now suppose $A$ has well-defined gramian matrices $\mG_p = \mP^\top \mP$ and $\mG_s = \mS^\top \mS$. This implies that the trace $\Tr(\mG_p \mG_s)$ is finite, which can be used to show that $f \in \ltworat$ as follows:
\begin{align}
\norm{f}_2^2 &= \sum_{x \in \sstar} f(x)^2 \leq \sum_{x \in \sstar} (|x| + 1) f(x)^2
= \Tr(\H_f \H_f^\top) \\
&= \Tr(\mP \mS^\top \mS \mP^\top) =
\Tr(\mP^\top \mP \mS^\top \mS) = \Tr(\mG_p \mG_s) < \infty \enspace.
\end{align}
\end{proof}

Note that the minimality assumption is not needed when showing that $A$ having well-defined gramians implies $f_A \in \ltworat$. On the other hand, the minimality of $A$ is essential to show that $f_A \in \ltworat$ implies that both gramians are well-defined, as witenessed by the example in Figure~\ref{fig:nogram}.

\begin{figure}[t]
\centering
\begin{subfigure}[c]{0.4\textwidth}
\centering
\begin{tikzpicture}[shorten >=1pt,node distance=3cm,auto]
\tikzstyle{every state}=[draw=blue!50,very thick,on grid,fill=blue!20,inner sep=3pt,minimum size=4mm]
\node[state with output,initial,initial text=$1$] (q_0) {$q_1$ \nodepart{lower} $1$};
\node[state with output] (q_1) [right of=q_0] {$q_2$ \nodepart{lower} $0$};
\path[->]
(q_0) edge [loop above] node
{\footnotesize $\begin{matrix}{a , 1/2}\end{matrix}$} ()
(q_0) edge [bend right] node [swap]
{\footnotesize $\begin{matrix}{b , 1/2}\end{matrix}$} (q_1)
(q_1) edge [loop above] node
{\footnotesize $\begin{matrix}{a , 1}\end{matrix}$} ();
\end{tikzpicture}
\end{subfigure}
\begin{subfigure}[c]{0.4\textwidth}
\centering
\begin{minipage}[t]{\textwidth}
\begin{tabular}{cc}
$\azero =
\left[
\begin{matrix}
1\\
0\\
\end{matrix}
\right]$
&
$\A_a =
\left[
\begin{matrix}
1/2 & 0 \\
0 & 1\\
\end{matrix}
\right]$\\[.5cm]
$\ainf =
\left[
\begin{matrix}
1\\
0\\
\end{matrix}
\right]$
&
$\A_b =
\left[
\begin{matrix}
0 & 1/2 \\
0 & 0\\
\end{matrix}
\right]$
\end{tabular}
\end{minipage}
\end{subfigure}
\caption{Example of non-minimal WFA $A$ computing a function in $\ltworat$ for which the forward Gramian is not defined. To see that $A$ is not minimal note that $f_A$ can be computed by the one state WFA obtained by removing $q_2$ from $A$. Note that $\mG_p (2,2)$ is not defined since $(\azero^\top \A_b \A_a^k \me_2)^2 = 1/4$ for all $k \geq 0$.}
\label{fig:nogram}
\end{figure}
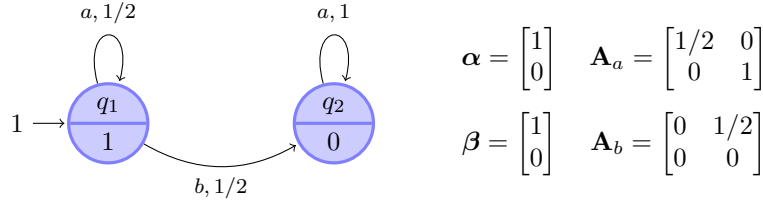

\subsection{Gramians and SVA}

The reason for introducing the reachability and observability gramians in the previous section is because these matrices can be used to reduce any given (minimal) WFA to its SVA form. The details of this construction are presented in this section, and they drawn upon some ideas already present in the proof of Proposition~\ref{prop:definedgramians}. Essentially, this section provides a reduction from the computation of the SVA to the computation of the gramians. The later problem is studied in detail in Section~\ref{sec:computegramians}.

Let $A$ be a minimal WFA with $n$ states realizing a function $f \in \ltworat$. By Proposition~\ref{prop:definedgramians} we know that the gramians $\mG_{A,p}$ and $\mG_{A,s}$ are defined. Furthermore, Theorems~\ref{thm:sva} and~\ref{thm:iffsvd} guarantee the existence of an invertible matrix $\mQ$ such that $B = A^{\mQ}$ is an SVA for $f$. Let $\mD$ be the diagonal matrix containing the singular values of the Hankel matrix of $f$. By inspecting the proof of Proposition~\ref{prop:definedgramians}, we see that these facts imply the following important equations:
\begin{align}
\mG_{B,p} &= \mD = \mQ^\top \mG_{A,p} \mQ \enspace, \label{eqn:gbp} \\
\mG_{B,s} &= \mD = \mQ^{-1} \mG_{A,s} \mQ^{-\top} \enspace. \label{eqn:gbs}
\end{align}
These equations say that given $A$ we can obtain its corresponding SVA by finding an invertible matrix $\mQ$ simultaneously transforming the Gramians of $A$ into two equal diagonal matrices. The following results provide a way to do this by taking the Cholesky decompositions of the Gramian matrices and computing an additional SVD.

\begin{lemma}\label{lem:diaggramsva}
Let $A$ be a minimal WFA with $n$ states realizing a function $f \in \ltworat$. Suppose the Gramians $\mG_p$ and $\mG_s$ satisfy $\mG_p = \mG_s = \mD = \diag(\sigma_1,\ldots, \sigma_n)$ with $\sigma_1 \geq \cdots \geq \sigma_n > 0$. Then $A$ is an SVA, and $\mD$ is the matrix of singular values of $\H_f$.
\end{lemma}
\begin{proof}
Let $\H_f = \mP \mS^\top$ be the FB factorization induced by $A$. Since $\mG_p = \mP^\top \mP$ and $\mG_s = \mS^\top \mS$ are diagonal and full rank, we see that the columns of $\mP$ (resp.\ $\mS$) are orthogonal. Now take $\mU = \mP \mD^{-1/2}$ and $\mV = \mS \mD^{-1/2}$ and note that these two matrices have orthonormal columns since $\mU^\top \mU = \mV^\top \mV = \mI$. Noting that $\H_f = \mP \mS^\top = \mU \mD \mV^\top$ is a decomposition satisfying the constraints of an SVD we conclude that $A$ is an SVA.
\end{proof}

\begin{theorem}
Let $A$ be a minimal WFA with $n$ states realizing a function $f \in \ltworat$ with Gramians $\mG_{s}$ and $\mG_{p}$. Let $\mG_s = \mL_s \mL_s^\top$ and $\mG_p = \mL_p \mL_p^\top$ be their Cholesky decompositions. Suppose $\mL_p^\top \mL_s$ has singular value decomposition $\mU \mD \mV^\top$. Then the WFA $B = A^{\mQ}$ with $\mQ = \mL_p^{-\top} \mU \mD^{1/2}$ is an SVA for $A$. Furthermore, we have $\mQ^{-1} = \mD^{1/2} \mV^\top \mL_s^{-1}$.
\end{theorem}
\begin{proof}
In the first place note that minimality of $A$ implies that $\mG_p$ and $\mG_s$ are full rank. Thus the factors $\mL_p$ and $\mL_s$ are invertible, $\mL_p^\top \mL_s$ has full rank, and both $\mQ$ and $\mQ^{-1}$ are well defined. To check the equality $\mQ^{-1} = \mD^{1/2} \mV^\top \mL_s^{-1}$ we just write
\begin{equation}
\left(\mL_p^{-\top} \mU \mD^{1/2}\right) \left(\mD^{1/2} \mV^\top \mL_s^{-1}\right) =
\mL_p^{-\top} \left(\mU \mD \mV^\top\right) \mL_s^{-1} = \mI \enspace.
\end{equation}
Next we check that $\mQ$ is such that $\mG_{B,p} = \mG_{B,s} = \mD$:
\begin{align*}
\mG_{B,p} &= \mQ^\top \mG_p \mQ = \left(\mD^{1/2} \mU^\top \mL_p^{-1}\right) \left(\mL_p \mL_p^\top\right) \left(\mL_p^{-\top} \mU \mD^{1/2}\right) = \mD \enspace, \\
\mG_{B,s} &= \mQ^{-1} \mG_s \mQ^{-\top} = \left(\mD^{1/2} \mV^\top \mL_s^{-1}\right) \left(\mL_s \mL_s^\top\right) \left(\mL_s^{-\top} \mV \mD^{1/2}\right) = \mD \enspace,
\end{align*}
where we used that $\mU^\top \mU = \mV^\top \mV = \mI$. Therefore, we can apply Lemma~\ref{lem:diaggramsva} to conclude that $B$ is an SVA.
\end{proof}

The previous theorem motivates the following simple algorithm for computing the SVA of a function $f \in \ltworat$ provided that a minimal WFA $A$ and its corresponding gramian matrices are given. We shall address the computation of the gramian matrices in the next section. For now we note that the constraint of $A$ being minimal is not essential, since its possible to minimize a WFA with $n$ states in time $O(n^3)$ \citep{berstel2011noncommutative}. Furthermore, given a minimal WFA $A$ it is possible to check the membership $f_A \in \ltworat$ using any of the tests discussed in Section~\ref{sec:ratfunswithsva}, which provides a way to verify the pre-condition necessary to ensure the existence of the gramian matrices.

\begin{algorithm}
\DontPrintSemicolon
\caption{\FuncSty{ComputeSVA}}\label{alg:sva}
\KwIn{A minimal WFA $A$ realizing $f \in \ltworat$, and the gramians $\mG_{A,p}$ and $\mG_{A,s}$}
\KwOut{An SVA $B$ for $f$}
Compute the Cholesky decompositions $\mG_s = \mL_s \mL_s^\top$ and $\mG_p = \mL_p \mL_p^\top$\;
Compute the SVD $\mU \mD \mV^\top$ of $\mL_p^\top \mL_s$\;
Let $B = A^{\mQ}$ with $\mQ = \mL_p^{-\top} \mU \mD^{1/2}$\;
\KwRet{$B$}\;
\end{algorithm}

The running time of $\FuncSty{ComputeSVA}(A)$ in terms of \emph{floating point operations} (flops) can be bounded using the following well-known facts about numerical linear algebra (see e.g.\ \citep{trefethen1997numerical}):
\begin{itemize}
\item Computing the product of two matrices $d \times d$ matrices takes time $O(d^3)$ if implemented naively, and can be done in time $O(d^{\omega})$ for some constant $\omega < 2.4$ using sophisticated algorithms that only yield practical improvements on very large matrices.
\item The singular value decomposition of a matrix $\mM \in \R^{d \times d}$ can be computed in time $O(d^3)$, and the Cholesy decomposition of a positive definite matrix $\mG \in \R^{d \times d}$ can also be computed in time $O(d^3)$.
\item The inverse of an invertible lower triangular matrix $\mL \in \R^{d \times d}$ can be computed in time $O(d^3)$ using Gaussian elimination.
\end{itemize}
Therefore, if the input $A$ to Algorithm~\ref{alg:sva} has $n$ states, its total running time is $O(n^3 + |\Sigma| n^\omega)$.

The following important observation about the product of two Gramians follows from the results showing how to compute an SVA from the Gramian matrices of a minimal WFA.

\begin{corollary}\label{cor:eigengpgs}
Let $A$ be a minimal WFA with $n$ states realizing a function $f \in \ltworat$. Then the product of the gramians $\mW = \mG_{A,s} \mG_{A,p}$ is a diagonalizable matrix with eigenvalues given by $\lambda_i(\mW) = \sigma_i(\H_f)^2$ for $i \in [n]$. Furthermore, if $\mQ$ is an invertible matrix such that $A^{\mQ}$ is an SVA, then $\mQ$ diagonalizes $\mW$; that is $\mW = \mQ \mD^2 \mQ^{-1}$.
\end{corollary}
\begin{proof}
Let $B = A^{\mQ}$ be an SVA for $f$ as above. By multiplying \eqref{eqn:gbp} and \eqref{eqn:gbs} together we see that
\begin{equation}
\mG_{B,s} \mG_{B,p} = \mD^2 = \mQ^{-1} \mG_{A,s} \mG_{A,p} \mQ = \mQ^{-1} \mW \mQ \enspace.
\end{equation}
Therefore, $\mW$ is diagonalizable and its eigenvalues are the squares of the Hankel singular values of $f$. Additionally, the above expression shows that $\mQ$ necessarily is a matrix of eigenvectors for $\mW$.
\end{proof}

\subsection{Gramian Fixed-Point Equations}\label{sec:fpeqns}

In addition to their definitions in terms of a FB factorization, the gramian matrices of a WFA can be characterized in terms of fixed-point equations. This point of view will prove useful later both for theoretical arguments as well as for developing algorithms for computing them.

\begin{theorem}\label{thm:fpeqs}
Let $A = \wa$ be a WFA with $n$ states such that the corresponding gramians $\mG_p$ and $\mG_s$ are well-defined. Then $\mX = \mG_p$ and $\mY = \mG_s$ are solutions to the following fixed-point equations:
\begin{align}
\mX &= \azero \azero^\top + \sum_{\symbola \in \Sigma} \A_\symbola^\top \mX \A_\symbola \enspace, \label{eqn:fpalpha} \\
\mY &= \ainf \ainf^\top + \sum_{\symbola \in \Sigma} \A_\symbola \mY \A_\symbola^\top \enspace. \label{eqn:fpbeta}
\end{align}
\end{theorem}
\begin{proof}
Recall that $\mG_p = \mP^\top \mP$ with $\mP \in \R^{\sstar \times n}$, and the row of $\mP$ corresponding to $x \in \sstar$ is given by $\azero^\top \A_x$. Expanding this definitions we get
\begin{align*}
\mG_p &= \sum_{x \in \sstar} (\A_x^\top \azero) (\azero^\top \A_x) \\
&= \azero \azero^\top + \sum_{x \in \Sigma^+} (\A_x^\top \azero) (\azero^\top \A_x) \\
&= \azero \azero^\top + \sum_{\symbola \in \Sigma} \sum_{x \in \sstar} \A_\symbola^\top (\A_x^\top \azero) (\azero^\top \A_x) \A_\symbola \\
&= \azero \azero^\top + \sum_{\symbola \in \Sigma} \A_\symbola^\top \left(\sum_{x \in \sstar} (\A_x^\top \azero) (\azero^\top \A_x) \right) \A_\symbola \enspace,
\end{align*}
where we just used that $\A_x^\top = (\A_{x_1} \cdots \A_{x_t})^\top = \A_{x_t}^\top \cdots \A_{x_1}^\top$ and that any string $y \in \Sigma^+$ satisfies $y = x \symbola$ for some $x \in \sstar$ and $\symbola \in \Sigma$. The derivation for $\mG_s$ follows exactly the same pattern.
\end{proof}

We note here that in the simple case where $|\Sigma| = 1$ equations \eqref{eqn:fpalpha} and \eqref{eqn:fpbeta} are special cases of the well-known \emph{discrete Lyapunov equation}.

Another important remark about this result is that the same argument used in the proof can be used to show that the matrices $\mG_p^{(t)}$ and $\mG_s^{(t)}$ defined in equations \eqref{eq:gpt} and \eqref{eq:gst} satisfy the following recurrence relations for any $t \geq 0$:
\begin{align}
\mG_p^{(t+1)} &= \azero \azero^\top + \sum_{\symbola \in \Sigma} \A_\symbola^\top \mG_p^{(t)} \A_\symbola \enspace, \label{eqn:recgp} \\
\mG_s^{(t+1)} &= \ainf \ainf^\top + \sum_{\symbola \in \Sigma} \A_\symbola \mG_s^{(t)} \A_\symbola^\top \enspace. \label{eqn:recgs}
\end{align}
Thus, for any WFA $A$ with $n$ states it will be convenient to define the mappings $F_p, F_s : \R^{n \times n} \to \R^{n \times n}$ given by
\begin{align}
F_p(\mX) &= \azero \azero^\top + \sum_{\symbola \in \Sigma} \A_\symbola^\top \mX \A_\symbola \enspace, \label{eqn:fp} \\
F_s(\mY) &= \ainf \ainf^\top + \sum_{\symbola \in \Sigma} \A_\symbola \mY \A_\symbola^\top \enspace. \label{eqn:fs}
\end{align}
With this notation, the results from this section can be summarized by saying that for any $t \geq 0$ we have $\mG_p^{(t)} = F_p^{t+1}(\mat{0})$, and when the Gramian $\mG_p$ is defined then it is a fixed point of the form $F_p(\mX) = \mX$ which can be obtained as $\lim_{t \to \infty} F_p^t(\mat{0})$. The same results apply to $\mG_s$ by replacing $F_p$ with $F_s$.

These maps satisfy an important property when applied to positive semi-definite matrices.

\begin{lemma}\label{lem:FpFsmonotone}
The maps $F_p$ and $F_s$ defined in \eqref{eqn:fp} and \eqref{eqn:fs} are monotonically increasing with respect to the Loewner order.
\end{lemma}
\begin{proof}
Let $\mX$ and $\mY$ be positive semi-definite matrices satisfying $\mX \geq \mY$. We need to show $F_p(\mX) \geq F_p(\mY)$. Recalling that for any matrices $\mM \geq \mat{0}$ and $\mQ$ one has $\mQ^\top \mM \mQ \geq \mat{0}$, we see that
\begin{equation}
F_p(\mX) - F_p(\mY) = \sum_a \A_a^\top (\mX - \mY) \A_a \geq 0 \enspace,
\end{equation}
since positive semi-definite matrices are closed under addition. The claim for $F_s$ follows from a similar argument.
\end{proof}

Finally, we conclude this section by stating a simple observation about the sequences $\mG_p^{(t)}$ and $\mG_s^{(t)}$ that will prove useful in the sequel.
 
\begin{lemma}\label{lem:gptpsd}
One has $\mG_p^{(t+1)} \geq \mG_p^{(t)}$ and $\mG_s^{(t+1)} \geq \mG_s^{(t)}$ for any $t$.
\end{lemma}
\begin{proof}
These just follow from \eqref{eq:gpt} and \eqref{eq:gst} by observing that the differences
\begin{align*}
\mG_p^{(t+1)} - \mG_p^{(t)} &= \sum_{|x| = t+1} \mP(x,:)^\top \mP(x,:) \enspace, \\
\mG_s^{(t+1)} - \mG_s^{(t)} &= \sum_{|x| = t+1} \mS(x,:)^\top \mS(x,:) \enspace,
\end{align*}
are positive semi-definite matrices.
\end{proof}

\subsection{Applications of Gramians}

We have seen so far that having the Gramians of a minimal WFA $A$ computing a rational function $f_A \in \ltworat$ is enough to efficiently find the SVA of $A$. We now show how having the Gramians of $A$ is also useful to compute several other quantities associated with $f_A$, including its $\ltwo$ norm.

\begin{theorem}\label{thm:gramiannorm}
Let $A = \wa$ be a WFA computing a rational function $f$. Then the following hold:
\begin{enumerate}
\item If the Gramian $\mG_s$ is defined then $\norm{f}_2^2 = \azero^\top \mG_s \azero$.
\item If the Gramian $\mG_p$ is defined then $\norm{f}_2^2 = \ainf^\top \mG_p \ainf$.
\item If both Gramians are defined then $\normop{\H_f}^2 = \rho(\mG_p \mG_s)$ and $\normstwo{\H_f}^2 = \Tr(\mG_p \mG_s)$.
\end{enumerate}
\end{theorem}
\begin{proof}
Suppose $\mG_s$ is defined. Letting $\bar{x}$ denote the reverse of a string $x$, the first equation follows from
\begin{align*}
\norm{f}_2^2 &= \sum_{x \in \sstar} f(x)^2 = \sum_{x \in \sstar} \left(\azero^\top \A_x \ainf\right) \left( \ainf^\top \A_{\bar{x}}^\top \azero\right) = \azero^\top \left(\sum_{x \in \sstar} \A_x \ainf \ainf^\top \A_{\bar{x}}^\top\right) \azero \\
&= \azero^\top \left(\sum_{x \in \sstar} \mS(x,:)^\top \mS(x,:) \right) \azero = \azero^\top \mG_s \azero \enspace.
\end{align*}
By writing $f(x)^2 = (\ainf^\top \A_{\bar{x}}^\top \azero) (\azero^\top \A_x \ainf)$, the proof of $\norm{f}_2^2 = \ainf^\top \mG_p \ainf$ follows from the same argument.

Now suppose both Gramians are defined and recall from Proposition~\ref{prop:definedgramians} that this implies $f \in \ltworat$. Therefore $\normop{H_f}$ and $\normstwo{H_f}$ are both finite. The desired equations follow directly from Corollary~\ref{cor:eigengpgs} by noting that $\rho(\mG_p \mG_s) = \lambda_1(\mG_p \mG_s) = \sigma_1(\H_f)^2 = \normop{\H_f}^2$ and $\Tr(\mG_p \mG_s) = \sum_{i=1}^n \lambda_i(\mG_p \mG_s) = \sum_{i=1}^n \sigma_i(\H_f)^2 = \normstwo{\H_f}^2$.

\end{proof}

Note that this last result shows that if either the reachability or observability Gramian of a possibly non-minimal WFA $A$ are defined, then we have $f_A \in \ltworat$. This gives a criterion for testing a WFA for finite $\ltwo$ norm in addition to those provided by Theorem~\ref{thm:l2sufficient}. It is also interesting to contrast this results with Proposition~\ref{prop:definedgramians}, in which we showed that if $A$ is minimal and $f_A \in \ltworat$, then both Gramians are necessarily defined.

\section{Computing the Gramians}\label{sec:computegramians}

In this section we present several algorithmic approaches for computing the SVA of a rational function in $\ltworat$ given in the form of an arbitrary minimal WFA. By Algorithm~\ref{alg:sva} this problem reduces to that of computing the Gramian matrices associated with the given WFA. The first approach works in the particular case where the fixed-point gramian equations have a unique solution, in which case the gramians can be efficiently computed by solving a system of linear equations. The second, more general algorithm is based on the computation of the least solution to a semi-definite system of matrix inequalities.

\subsection{The Unique Solution Case}

The main idea behind our first algorithm for computing the gramian matrices of a WFA is based on directly exploiting the definitions of these matrices. In particular, since $\mG_p = \mP^\top \mP$, we have that $\mG_p(i,j)$ is the inner product between the $i$th and the $j$th columns of $\mP$. By noting that each of these columns is in fact a rational function, we see that computing $\mG_p$ can be reduced to the problem of computing the inner product of two rational functions. Since it is possible to find a closed-form solution to this inner product computation, this observation can be exploited to compute $\mG_p$ directly by obtaining these inner products one at at time. However, we will observe that a significant amount of these calculations can actually be reused from entry to entry. This motivates the development of an improved procedure that efficiently exploits this structure by amortizing the shared computations among all entries in $\mG_p$. Of course, by symmetry the very same arguments can be applied to the gramian $\mG_s$.

We start with the following simple observation about solutions to the gramian fixed-point equations.

\begin{lemma}\label{lem:equiveqs}
Let $A = \wa$ be a WFA with $n$ states and $\mX \in \R^{n \times n}$ an arbitrary matrix. Recall that $\vx = \vecm(\mX) \in \R^{n^2}$ is the vector obtained by concatenating the columns of $\mX$. Then the following hold:
\begin{enumerate}
\item $\mX$ is a solution of $\mX = \azero \azero^\top + \sum_{\symbola} \A_\symbola^\top \mX \A_\symbola$ if and only if $\vx$ is a solution of $(\azero \otimes \azero)^\top = \vx^\top (\mI - \sum_\symbola \A_\symbola \otimes \A_\symbola)$,
\item $\mX$ is a solution of $\mX = \ainf \ainf^\top + \sum_{\symbola} \A_\symbola \mX \A_\symbola^\top$ if and only if $\vx$ is a solution of $(\ainf \otimes \ainf) = (\mI - \sum_\symbola \A_\symbola \otimes \A_\symbola) \vx$.
\end{enumerate}
\end{lemma}
\begin{proof}
The result follows immediately from the well-known relations $\vecm(\v \v^\top) = \v \otimes \v$ and $\vecm(\A \mX \B^\top) = (\B \otimes \A) \vecm(\mX)$, and the linearity of the $\vecm(\bullet)$ operation.
\end{proof}

Now we can show that the fixed-point equations have a unique solution when a simple condition is satisfied. This yields an efficient algorithm for computing $\mG_p$ and $\mG_s$ when an easily testable condition holds.

\begin{theorem}\label{thm:uniquefp}
Let $A = \wa$ be a WFA with $n$ states and denote by $\rho$ the spectral radius of the matrix $\sum_\symbola \A_\symbola \otimes \A_\symbola$. If $\rho < 1$ then the following are satisfied:
\begin{enumerate}
\item $\vx = \vecm(\mG_p)$ is the unique solution to $(\azero \otimes \azero)^\top = \vx^\top (\mI - \sum_\symbola \A_\symbola \otimes \A_\symbola)$
\item $\vy = \vecm(\mG_s)$ is the unique solution to $(\ainf \otimes \ainf) = (\mI - \sum_\symbola \A_\symbola \otimes \A_\symbola) \vy$
\end{enumerate}
\end{theorem}
\begin{proof}
Recall that the WFA $B = \langle \azero \otimes \azero, \ainf \otimes \ainf, \{\A_\symbola \otimes \A_\symbola\} \rangle$ satisfies $f_B = f_A^2$. Therefore we have $f_B(x) \geq 0$ for all $x \in \sstar$. Using the assumption on $\rho$ and Lemma~\ref{lem:denis} we see that $f_B \in \lonerat$ and therefore $f_A \in \ltworat$, which by Proposition~\ref{prop:definedgramians} implies that $\mG_{A,p}$ and $\mG_{A,s}$ are well-defined. Therefore Theorem~\ref{thm:fpeqs} and  Lemma~\ref{lem:equiveqs} tell us that both $(\azero \otimes \azero)^\top = \vx^\top (\mI - \sum_\symbola \A_\symbola \otimes \A_\symbola)$ and $(\ainf \otimes \ainf) = (\mI - \sum_\symbola \A_\symbola \otimes \A_\symbola) \vy$ have at least one solution.

Suppose $\vy, \vy' \in \R^{n^2}$ are two solutions to equation $(\ainf \otimes \ainf) = (\mI - \sum_\symbola \A_\symbola \otimes \A_\symbola) \vy$. This implies that $(\mI - \sum_\symbola \A_\symbola \otimes \A_\symbola) \vy = (\mI - \sum_\symbola \A_\symbola \otimes \A_\symbola) \vy'$, from where we deduce that $\vy - \vy' = (\sum_\symbola \A_\symbola \otimes \A_\symbola) (\vy - \vy')$. Thus, either $\vy = \vy'$ or $\vy - \vy'$ is an eigenvector of $\sum_\symbola \A_\symbola \otimes \A_\symbola$ with eigenvalue $1$. Since the latter is not possible because we assumed $\rho < 1$, we conclude that the solution is unique. The same argument applies to $(\azero \otimes \azero)^\top = \vx^\top (\mI - \sum_\symbola \A_\symbola \otimes \A_\symbola)$.
\end{proof}

\subsection{The General Case}

In the case where the automaton $A = \wa$ is such that $\lambda = 1$ is an eigenvalue of $\sum_a \A_a \otimes \A_a$, then the linear system considered in the previous section will not have a unique solution. For example, this might occur when $A$ is minimal but $A \otimes A$ is not. Therefore, in general we will need some extra information about the gramian matrices in order to find them among the subset of possible solutions of the linear systems given by Lemma~\ref{lem:equiveqs}. This information is provided by our next lemma, which states that the gramian matrices are the least positive-semidefinite solutions of some linear matrix inequalities. Throughout this section we assume that $A = \wa$ is a WFA with $n$ states such that the corresponding gramians $\mG_p$ and $\mG_s$ are well-defined, and therefore the linear systems in Lemma~\ref{lem:equiveqs} admit at least one solution.

\begin{lemma}\label{lem:psdlmi}
The following hold:
\begin{enumerate}
\item The Gramian $\mG_p$ is the least positive semi-definite solution to the linear matrix inequality
\begin{equation}
\mX \geq \azero \azero^\top + \sum_{\symbola \in \Sigma} \A_\symbola^\top \mX \A_\symbola \enspace. \label{eqn:psdalpha}
\end{equation}
\item The Gramian $\mG_s$ is the least positive semi-definite solution to the linear matrix inequality
\begin{equation}
\mY \geq \ainf \ainf^\top + \sum_{\symbola \in \Sigma} \A_\symbola \mY \A_\symbola^\top \enspace. \label{eqn:psdbeta}
\end{equation}
\end{enumerate}
\end{lemma}
\begin{proof}
Since the proofs of both statements follow exactly the same structure, we give only the proof for $\mG_p$. From Theorem~\ref{thm:fpeqs} it follows that $\mG_p$ satisfies \eqref{eqn:psdalpha}. Now let $\mX$ be another positive semi-definite matrix satisfying \eqref{eqn:psdalpha}. We will show by induction that for every $t \geq 0$ we have
\begin{equation}
\mX \geq \sum_{x \in \Sigma^{\leq t}} \A_x^\top \azero \azero^\top \A_x + \sum_{x \in \Sigma^{t+1}} \A_x^\top \mX \A_x \enspace.
\label{eqn:recursionX}
\end{equation}
First note that the case $t = 0$ follows immediately from \eqref{eqn:psdalpha}. Now assume the inequality is true for some $t$ and consider the case $t+1$. We have
\begin{align*}
\mX
&\geq \sum_{x \in \Sigma^{\leq t}} \A_x^\top \azero \azero^\top \A_x + \sum_{x \in \Sigma^{t+1}} \A_x^\top \mX \A_x \\
&\geq \sum_{x \in \Sigma^{\leq t}} \A_x^\top \azero \azero^\top \A_x + \sum_{x \in \Sigma^{t+1}} \A_x^\top \azero \azero^\top \A_x + \sum_{x \in \Sigma^{t+2}} \A_x^\top \mX \A_x \\
&= \sum_{x \in \Sigma^{\leq t+1}} \A_x^\top \azero \azero^\top \A_x + \sum_{x \in \Sigma^{t+2}} \A_x^\top \mX \A_x \enspace,
\end{align*}
where the second inequality uses \eqref{eqn:psdalpha} and the fact that $\mY \geq \mZ$ implies $\mM^\top \mY \mM \geq \mM^\top \mZ \mM$ for any matrix $\mM$. By rewriting \eqref{eqn:recursionX} and noting that $\sum_{x \in \Sigma^{t+1}} \A_x^\top \mX \A_x \geq \mat{0}$ for any $t \geq 0$, we see that
\begin{equation}
\sum_{x \in \Sigma^{\leq t}} \A_x^\top \azero \azero^\top \A_x \leq \mX - \sum_{x \in \Sigma^{t+1}} \A_x^\top \mX \A_x \leq \mX \enspace.
\end{equation}
Since $\mG_p$ is defined we must have $\mG_p = \lim_{t \to \infty} \sum_{x \in \Sigma^{\leq t}} \A_x^\top \azero \azero^\top \A_x$, and therefore $\mG_p \leq \mX$.
\end{proof}

As a direct consequence of the above lemma we get the following characterization for the Gramian matrices of any WFA $A$ with $f_A \in \ltworat$.

\begin{theorem}\label{thm:gramlfp}
The Gramian $\mG_p$ (resp.\ $\mG_s$) is the least positive semi-definite fixed point of \eqref{eqn:fpalpha} (resp.\ \eqref{eqn:fpbeta}).
\end{theorem}
\begin{proof}
For $\mG_p$ the result follows from Lemma~\ref{lem:psdlmi} since any fixed-point of \eqref{eqn:fpalpha} satisfies \eqref{eqn:psdalpha}; the same holds for $\mG_s$.
\end{proof}

Using this characterization we can derive an efficient algorithm for finding the Gramian matrices even when the linear systems given by Lemma~\ref{lem:equiveqs} have more than one solution. The solution is based on solving a semi-definite optimization program. For simplicity we only present the optimization problem for finding the Gramian $\mG_s$ and note that a completely symmetric argument also works for $\mG_p$. We start by introducing some notation. Let $\mM = \mI + \sum_{a} \A_a \otimes \A_a \in \R^{n^2 \times n^2}$ and $\vy_0 = \mM^{\dagger} (\ainf \otimes \ainf)$. Also, let $\vy_1, \ldots, \vy_d \in \R^{n^2}$ be a basis of linearly independent vectors for the column-space of the matrix $\mI - \mM^{\dagger} \mM$. For $0 \leq i \leq d$ we write $\mY_i \in \R^{n \times n}$ to denote the matrix such that $\vy_i = \vecm(\mY_i)$. Finally, we let $\pi$ denote the linear map representing the orthogonal projection onto the space of $n \times n$ symmetric matrices, which is given by $\pi(\mY) = (\mY + \mY^\top)/2$. With this notation we define the following semi-definite optimization problem:
\begin{alignat}{3}
& \underset{t_1, \ldots, t_d \in \R}{\text{minimize}}
&\enspace \enspace & \sum_{i=1}^d t_i \Tr(\mY_i) \label{eqn:sdp} \\
& \text{subject to}
& & \pi(\mY_0) + \sum_{i=1}^d t_i \pi(\mY_i) \geq \mat{0} \enspace , \label{eqn:sdpc1} \\
&&& \mY_0 - \mY_0^\top + \sum_{i=1}^d t_i (\mY_i - \mY_i^\top) = \mat{0} \enspace. \label{eqn:sdpc2}
\end{alignat}

\begin{theorem}
Let $t_1^*, \ldots, t_d^*$ be the optimal solution to \eqref{eqn:sdp}. Then the matrix $\mY^* = \mY_0 + \sum_{i=1}^d t_i^* \mY_i$ is the least positive semi-definite solution of $\mY = \ainf \ainf^\top + \sum_{a} \A_a \mY \A_a^\top$.
\end{theorem}
\begin{proof}
We start by observing that all solutions to $\mY = \ainf \ainf^\top + \sum_{a} \A_a \mY \A_a^\top$ are of the form $\mY = \mY_0 + \sum_{i=1}^d t_i \mY_i$ for some $t_1, \ldots, t_d$. This follows from the fact that the Moore-Penrose pseudo-inverse can be used to show that every solution of the linear system $\ainf \otimes \ainf = (\mI - \sum_{a} \A_a \otimes \A_a) \vy$ can be written in the form $\mM^{\dagger} (\ainf \otimes \ainf) + (\mI - \mM^{\dagger} \mM) \vz$ for some $\vz \in \R^{n^2}$. Since any solution of this form can be rewritten as $\vy_0 + \sum_{i=1}^d t_i \vy_i$, the claim follows directly by the linearity of the $\vecm(\cdot)$ operation.

Next we show that any matrix of the form $\mY = \mY_0 + \sum_{i=1}^d t_i \mY_i$ satisfying \eqref{eqn:sdpc1} and \eqref{eqn:sdpc2} is symmetric and positive semi-definite. Indeed, if \eqref{eqn:sdpc2} is satisfied then $\pi(\mY) = \mY$ since
\begin{align*}
\mY - \pi(\mY) &= \left(\mY_0 - \frac{\mY_0 + \mY_0^\top}{2}\right) + \sum_{i=1}^d t_i \left(\mY_i - \frac{\mY_i + \mY_i^\top}{2}\right) \\
&= \frac{\mY_0 - \mY_0^\top}{2} + \sum_{i=1}^d t_i \frac{\mY_i - \mY_i^\top}{2} = \mat{0} \enspace.
\end{align*}
Therefore $\mY$ is symmetric and \eqref{eqn:sdpc1} implies $\mY = \pi(\mY) \geq \mat{0}$, so $\mY$ is positive semi-definite.

Finally suppose $\mY$ and $\mY'$ are two positive semi-definite solutions of \eqref{eqn:fpbeta} with $\mY \leq \mY'$. Then by linearity of the trace we have $\Tr(\mY) \leq \Tr(\mY')$. Therefore, the least positive semi-definite solution to \eqref{eqn:fpbeta} is also the positive semi-definite solution with minimum trace $\mY^*$ obtained by solving \eqref{eqn:sdp}.
\end{proof}

\section{Application: Approximate Minimization of WFA}\label{sec:approxmin}

The fact that given a (minimal) WFA realizing a function in $\ltworat$ we can efficiently compute its corresponding SVA opens the door to multiple applications. In this section we focus on the application of SVA to the design and analysis of algorithms for model reduction. To motivate the need for such algorithms, consider the situation where one has a WFA modelling a system of interest and the need arises for testing whether the system satisfies a given property. If testing this property requires multiple evaluations of the function computed by the WFA, the cost of this computation will grow with the number of states, and if the system is large the repeated evaluation of millions of queries might take a very long time. But if the decision about the property being satisfied does not depend too much on the individual answers of each query, it might be acceptable to provide \emph{approximate} answers for each of these queries. If in addition these approximate queries can be performed much faster than exact queries, then the whole testing process can be sped up by trading-off accuracy and query processing time. In the rest of this section we formalise the problem of approximate evaluations of WFA, and provide a solution based on the truncation of SVA canonical forms.

\subsection{Problem Formulation}

We now proceed to give a formal definition of the approximate minimization problem for WFA. Roughly speaking, this corresponds to finding a small WFA computing a good approximation to the function realized by a large minimal WFA. A solution to this problem will yield a way to speed up approximate evaluation of WFA.

Let $A$ be a minimal WFA with $n$ states computing a rational function $f \in \ltworat$. Given a target number of states $\hat{n} < n$, we want to find a WFA $\hat{A}$ with $\hat{n}$ states computing a function $\hat{f}$ which minimizes $\norm{f - \hat{f}}_2$ among all rational function of rank at most $\hat{n}$. This problem can be formulated as an optimization problem as follows:
\begin{equation}\label{eqn:optapprmin}
\inf_{\rank(\hat{f}) \leq \hat{n}} \norm{f - \hat{f}}_2 \enspace.
\end{equation}
The first observation we make about this problem is that, although it is not explicitly encoded in \eqref{eqn:optapprmin}, any solution $\hat{f}$ will have finite $\ltwo$ norm. Indeed, it is easy to see that
\begin{equation}
\norm{\hat{f}}_2 \leq \norm{f}_2 + \norm{f - \hat{f}}_2 =
\norm{f}_2 + \inf_{\rank(\hat{f}) \leq \hat{n}} \norm{f - \hat{f}}_2 \leq 2 \norm{f}_2 \enspace,
\end{equation}
where the last inequality uses that the rational function $0$ has rank $1$ and therefore it is a feasible point of the optimization \eqref{eqn:optapprmin}.

The second important observation is that, like rank constrained optimizations over finite matrices, the optimization in \eqref{eqn:optapprmin} is not convex. To see this, let us write $A = \wa$ for the original automaton and $\hat{A} = \hwa$ for the automaton we are looking for, noting that any automaton with at most $\hat{n}$ states can be written as a (non-minimal) WFA with $\hat{n}$. Then, by the motonicity of $z \mapsto z^2$ we can replace the objective $\norm{f - \hat{f}}_2$ with $\norm{f - \hat{f}}_2^2$, and see that, using the WFA representation for $f - \hat{f}$ and the closed-form expression for $\norm{f - \hat{f}}_2^2$ in terms of this WFA representation, \eqref{eqn:optapprmin} can be rewritten as the minimization over $\hat{A}$ of the quantity
\begin{equation}
[\azero^\top \; \hazero^\top] \otimes [\azero^\top \; \hazero^\top]
\left( \mI - \sum_{a \in \Sigma}
\left[\begin{array}{cc}
\A_a & \mat{0} \\
\mat{0} & \hat{\A}_a
\end{array}
\right]
\otimes
\left[\begin{array}{cc}
\A_a & \mat{0} \\
\mat{0} & \hat{\A}_a
\end{array}
\right]
\right)^{-1}
\left[\begin{array}{c}
\ainf \\
-\hainf
\end{array}\right]
\otimes
\left[\begin{array}{c}
\ainf \\
-\hainf
\end{array}\right] \enspace.
\end{equation}
Since this equivalent objective function is not convex, we have little hope of being able to efficiently solve \eqref{eqn:optapprmin} exactly. Instead, we will take a different approach and see how truncating the SVA for $A$ to have $\hat{n}$ states yields an approximate solution which can be efficiently computed.

\subsection{SVA Truncation}

In this section we describe an approximate minimization algorithm for WFA realizing a function in $\ltworat$. The algorithm takes as input a minimal WFA $A$ with $n$ states and a target number of states $\hat{n}$, and outputs a new WFA $\hat{A}$ with $\hat{n}$ states approximating the original WFA $A$.
To obtain $\hat{A}$ we first compute the SVA $A'$ associated to $A$, and then remove the $n - \hat{n}$ states associated with the smallest singular values of $\H_{f_A}$. More formally, by writing the block decomposition of the operators associated with the SVA $A'$ shown below, we get the operators for $\hat{A}$ by taking the sub-block in the top left containing the first $\hat{n}$ rows and $\hat{n}$ columns:
\begin{equation}\label{eqn:svablock}
\A'_a =
\left[
\begin{array}{cc}
\A_a^{(11)} & \A_a^{(12)} \\
\A_a^{(21)} & \A_a^{(22)}
\end{array}
\right] \enspace,
\qquad
\hat{\A}_a = \left[ \A_a^{(11)} \right] \enspace.
\end{equation}
Note that if we define the matrix $\mGamma = [\mat{I}_{\hat{n}} \; \mat{0}] \in \R^{\hat{n} \times n}$, then we have $\hat{\A}_a = \mGamma \A'_\symbola \mGamma^\top$. To reflect this fact we shall sometimes write $\hat{A} = \mGamma A' \mGamma^\top$.
Algorithm~\ref{alg:minsva} provides a description of the full procedure, which we call \FuncSty{SVATruncation}.
Since the algorithm only involves a call to \FuncSty{ComputeSVA} and a simple algebraic manipulation of the resulting WFA, the running time of \FuncSty{SVATruncation} is dominated by the complexity of \FuncSty{ComputeSVA}, and hence is polynomial in $|\Sigma|$, $\dim(A)$ and $\hat{n}$.

\begin{algorithm}
\DontPrintSemicolon
\caption{\FuncSty{SVATruncation}}\label{alg:minsva}
\KwIn{A minimal WFA $A$ with $n$ states, a target number of states $\hat{n} < n$}
\KwOut{A WFA $\hat{A}$ with $\hat{n}$ states}
Let $A' \leftarrow \FuncSty{ComputeSVA}(A)$\;
Let $\mat{\Gamma} = [\mat{I}_{\hat{n}} \; \mat{0}] \in \R^{\hat{n} \times n}$\;
Let $\hA_\symbola = \mat{\Gamma} \A'_\symbola \mat{\Gamma}^\top$ for all $\symbola \in
\Sigma$\;
Let $\hazero = \mat{\Gamma} \azero'$\;
Let $\hainf = \mat{\Gamma} \ainf'$\;
Let $\hat{A} = \hwa$\;
\KwRet{$\hat{A}$}\;
\end{algorithm}

Roughly speaking, the rationale behind \FuncSty{SVATruncation} is that given an SVA, the states corresponding to the smallest singular values are the ones with less influence on the Hankel matrix, and therefore should also be the ones with less influence on the associated rational function.
However, the details are more tricky than this simple intuition. The reason being that a low rank approximation to $\H_f$ obtained by truncating its SVD is not in general a Hankel matrix, and therefore does not correspond to any rational function. In particular, the Hankel matrix of the function $\hat{f}$ computed by $\hat{A}$ is not obtained by truncating the SVD of $\H_f$. This makes our analysis more involved than just applying the well-known bounds for low-rank approximation based on SVD. Nonetheless, we are able to obtain a bound of the same form that one would expect by measuring the error of a low-rank approximation using the Frobenius norm. Along these lines, our main result is the following theorem, which bounds the $\ltwo$-distance between the rational function $f$ realized by the original WFA $A$, and the rational function $\hat{f}$ realized by the output WFA $\hat{A}$.

\begin{theorem}\label{thm:nicebound}
Let $A$ be a minimal WFA with $n$ states computing a function $f \in \ltworat$ with Hankel singular values $\sigma_1 \geq \cdots \geq \sigma_n > 0$. Let $\hat{f}$ denote the function computed by the truncated SVA $\hat{A}$ with $1 \leq \hat{n} < n$ states. Then the following holds:
\begin{equation}
\label{eq:svatrunc}
\norm{f - \hat{f}}_2^2 \leq \sum_{i = \hat{n}+1}^{n} \sval_i^2 \enspace.
\end{equation}
\end{theorem}

The proof will be given in Section~\ref{sec:svaboundproof}. First, a few remarks about this result are in order.
The first is to observe that because $\sval_1 \geq \cdots \geq \sval_n$, the error decreases when $\hat{n}$ increases, which is the desired behavior: the more states $\hat{A}$ has, the closer it is to $A$.
The second is that \eqref{eq:svatrunc} does not depend on which representation $A$ of $f$ is given as input to \FuncSty{SVATruncation}. This is a consequence of first obtaining the corresponding SVA $A'$ before truncating.
Obviously, one could obtain another approximate minimization by truncating $A$ directly. However, in that case the final error would depend on the initial $A$ and in general it does not seem possible to use this approach for providing \emph{representation independent} bounds on the quality of approximation.
To see the importance of starting the truncation procedure from the SVA canonical form let us consider the following result which follows directly from the gramian fixed-point equations for SVA.

\begin{lemma}\label{lem:svaequations}
Let $A = \wa$ be an SVA with $n$ states realizing a function $f \in \ltworat$ with Hankel singular values $\sval_1 \geq \cdots \geq \sval_n$.
Then the following are satisfied:
\begin{enumerate}
\item For all $j \in [n]$, $\sum_{i} \sval_i \sum_{\symbola} \A_{\symbola}(i,j)^2 =
\sval_j - \azero(j)^2$,
\item For all $i \in [n]$, $\sum_{j} \sval_j \sum_{\symbola} \A_{\symbola}(i,j)^2 =
\sval_i - \ainf(i)^2$.
\end{enumerate}
\end{lemma}
\begin{proof}
These equations correspond to the diagonal entries of the gramian fixed-point equations for SVA
\begin{align}
\mD &= \azero \azero^\top + \sum_{\symbola} \A_\symbola^\top \mD \A_\symbola \enspace, \\
\mD &= \ainf \ainf^\top + \sum_{\symbola} \A_\symbola \mD \A_\symbola^\top \enspace.
\end{align}
\end{proof}

To see why this lemma justifies the truncation of an SVA we consider the following simple consequence.
By fixing $i, j \in [n]$ and $\symbola \in \Sigma$, we can use the first equation to get
\begin{align*}
\sval_i \A_{\symbola}(i,j)^2 = \sval_j - \azero(j)^2 - \left( \sum_{i} \sval_i \sum_{\symbola} \A_{\symbola}(i,j)^2 - \sval_i \A_{\symbola}(i,j)^2 \right) \leq \sval_j \enspace.
\end{align*}
Applying a similar argument to the second equation, we conclude that
\begin{equation*}
|\A_\symbola(i,j)| \leq \min\left\{ \sqrt{\frac{\sval_i}{\sval_j}},
\sqrt{\frac{\sval_j}{\sval_i}} \right\}
= \sqrt{\frac{\min\{\sval_i,\sval_j\}}{\max\{\sval_i,\sval_j\}}} \enspace.
\end{equation*}
This bound is telling us that in an SVA, transition weights further away from the diagonals of the $\A_\symbola$ are going to be small whenever there is a wide spread between the largest and smallest singular values; for example, $|\A_\symbola(1,n)| \leq \sqrt{\sval_n/\sval_1}$.
Intuitively, this means that in an SVA the last states are very weakly connected to the first states, and therefore removing these connections should not affect the output of the WFA too much.
The proof of Theorem~\ref{thm:nicebound} exploits this intuition, while at the same time leverages the full power of the fixed-point SVA Gramian equations.

We finish this section by stating another result about $\hat{A}$: SVA truncation always reduces the norm of the original function. Logically speaking, this is a preliminary to Theorem~\ref{thm:nicebound}, since it shows that the function computed by $\hat{A}$ has finite $\ltwo$ norm and already implies the finiteness of $\norm{f - \hat{f}}_2$. From an approximation point of view, this result basically says that SVA truncation can be interpreted as an algorithm for approximate minimization ``from below'', which might be a desirable property in some applications.

\begin{theorem}\label{thm:normfhat}
Let $A$ be a WFA computing a function $f \in \ltworat$ of rank $n$ and $\hat{A}$ a truncation of the SVA of $A$ with $\hat{n} < n$ states. The function $\hat{f}$ computed by $\hat{A}$ satisfies $\hat{f} \in \ltworat$ and $\norm{\hat{f}}_2 \leq \norm{f}_2$.
\end{theorem}

It is important to note that in general truncating an arbitrary WFA does not always reduce its norm as shown by the example in Figure~\ref{fig:noreduc}.

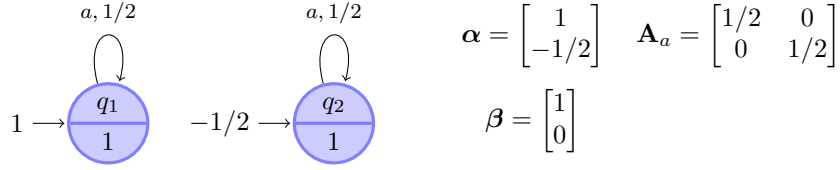
\begin{figure}[t]
\centering
\begin{subfigure}[c]{0.4\textwidth}
\centering
\begin{tikzpicture}[shorten >=1pt,node distance=3cm,auto]
\tikzstyle{every state}=[draw=blue!50,very thick,on grid,fill=blue!20,inner sep=3pt,minimum size=4mm]
\node[state with output,initial,initial text=$1$] (q_0) {$q_1$ \nodepart{lower} $1$};
\node[state with output,initial,initial text=$-1/2$] (q_1) [right of=q_0] {$q_2$ \nodepart{lower} $1$};
\path[->]
(q_0) edge [loop above] node
{\footnotesize $\begin{matrix}{a , 1/2}\end{matrix}$} ()
(q_1) edge [loop above] node
{\footnotesize $\begin{matrix}{a , 1/2}\end{matrix}$} ();
\end{tikzpicture}
\end{subfigure}
\begin{subfigure}[c]{0.4\textwidth}
\centering
\begin{minipage}[t]{\textwidth}
\begin{tabular}{cc}
$\azero =
\left[
\begin{matrix}
1\\
-1/2\\
\end{matrix}
\right]$
&
$\A_a =
\left[
\begin{matrix}
1/2 & 0 \\
0 & 1/2\\
\end{matrix}
\right]$\\[.5cm]
$\ainf =
\left[
\begin{matrix}
1\\
0\\
\end{matrix}
\right]$
&
\end{tabular}
\end{minipage}
\end{subfigure}
\caption{Example of WFA $A$ such that $\norm{f_{\hat{A}}}_2 \leq \norm{f_A}_2$, where $\hat{A}$ is the automaton obtained by removing the state $q_2$. In particular, $\norm{f_A}_2^2 = 1/3$ and $\norm{f_{\hat{A}}}_2^2 = 4/3$.}
\label{fig:noreduc}
\end{figure}

\subsection{SVA Truncation: Bounding the Norm}

The proof of the Theorem~\ref{thm:normfhat} illustrates how having different ways to represent the function $\hat{f}$ computed by the SVA truncation can be useful; this fact will also be essential in the proof of Theorem~\ref{thm:nicebound}. In particular, given an SVA $A$, we note that the automaton $\tilde{A}$ obtained by padding with zeros all the coefficients in the initial and transition weights of $A$ that are removed when taking its truncation in \FuncSty{SVATruncation} computes the same function as $\hat{A}$.

More concretely, let us recall the notation from \eqref{eqn:svablock} splitting of the weights conforming $A$ into a block corresponding to states $1$ to $\hat{n}$, and another block containing states $\hat{n}+1$ to $n$. We can define a similar partition for the initial and final weights of $A$. In particular, we write the following:
\begin{align*}
\azero &= \left[ \begin{array}{c} \azero^{(1)} \\ \azero^{(2)} \end{array} \right] \enspace, \\
\ainf &= \left[ \begin{array}{c} \ainf^{(1)} \\ \ainf^{(2)} \end{array} \right]
\enspace, \\
\A_\symbola &= \left[ \begin{array}{cc} \A_\symbola^{(11)} & \A_\symbola^{(12)} \\
\A_\symbola^{(21)} & \A_\symbola^{(22)} \end{array} \right] \enspace.
\end{align*}
Now the SVA truncation $\hat{A} = \mGamma A \mGamma^\top = \hwa$ with $\mGamma = [\mat{I}_{\hat{n}} \; \mat{0}]$ can be written in terms of this block decomposition as $\hazero = \azero^{(1)}$, $\hainf = \ainf^{(1)}$, and $\hA_\symbola = \A_\symbola^{(11)}$.

The important observation here is that starting from $A$ we can write other WFA computing the same function as $\hat{A}$. The following construction yields a WFA $\tilde{A}$ of size $n$ with this property. Define the $n \times n$ matrix
\begin{equation}
\mPi = \left[
\begin{array}{cc}
\mI_{\hat{n}} & \mat{0} \\
\mat{0} & \mat{0}
\end{array}
\right]
\end{equation}
and let $\tilde{A} = \twa$ with $\tazero = \mPi \azero$, $\tainf = \ainf$, and $\tilde{\A}_a = \A_a \mPi^\top = \A_a \mPi$. For convenience we shall sometimes write $\tilde{A} = A \mPi$. Note that the weights of $\tilde{A}$ are given by:
\begin{align*}
\tazero &= \left[ \begin{array}{c} \azero^{(1)} \\ \mat{0} \end{array} \right] \enspace, \\
\tainf &= \left[ \begin{array}{c} \ainf^{(1)} \\ \ainf^{(2)} \end{array} \right]
\enspace, \\
\tA_\symbola &= \left[ \begin{array}{cc} \A_\symbola^{(11)} & 
\mat{0} \\
\A_\symbola^{(21)} & \mat{0} \end{array} \right] \enspace.
\end{align*}

\begin{lemma}\label{lem:equivsvatrunc}
Let $A = \wa$ be an SVA with $n$ states. If $\hat{A} = \mGamma A \mGamma^\top = \hwa$ is the truncation of $A$ with $\hat{n}$ states, and $\tilde{A} = A \mPi = \twa$ is the the WFA with $n$ states defined above, then $\hat{A}$ and $\tilde{A}$ compute the same function $\hat{f}$.
\end{lemma}
\begin{proof}
Given $x \in \sstar$ define $\tazero_x^\top = \tazero^\top \tA_x$ and $\hazero_x^\top = \hazero^\top \hA_x$. By using the pattern of zeros in $\tA_x$, a simple induction argument on the length of $x$ shows that the following is always satisfied:
\begin{equation}
\tazero_x^\top = [\hazero_x^\top \; \mat{0}] \enspace.
\end{equation}
Therefore for any $x \in \sstar$ we have $f_{\tilde{A}}(x) = \tazero_x^\top \tainf = \hazero_x^\top \tainf^{(1)} = f_{\hat{A}}(x)$.
\end{proof}

The advantage of having a WFA with $n$ states computing the same function as the SVA truncation is that now both $A$ and $\tilde{A}$ have Gramians of the same dimensions which can be compared. The following lemma provides such comparison.

\begin{lemma}\label{lem:gramstildeA}
Let $A$ be an SVA with $n$ states and reachability Gramian $\mG_{p} = \mD$. Let $\tilde{A} = A \mPi$ the WFA with $n$ states computing the same function as the truncation of $A$ with $\hat{n}$ states. Then the reachability Gramian $\tilde{\mG}_p$ of $\tilde{A}$ is defined and satisfies $\mG_{p} \geq \tilde{\mG}_{p}$.
\end{lemma}
\begin{proof}
Recall the definition of the map $F_p$ for $A$ from Section~\ref{sec:fpeqns} and note that the corresponding map for $\tilde{A}$ satisfies
\begin{equation}
\tilde{F}_p(\mX) = \tazero \tazero^\top + \sum_a \tA_a^\top \mX \tA_a = \mPi F_p(\mX) \mPi \enspace.
\end{equation}
Taking $\tilde{\mG}^{(0)}_{p} = \mat{0}$ and $\tilde{\mG}^{(t+1)}_{p} = \tilde{F}_p(\tilde{\mG}^{(t)}_{p})$ we have $\tilde{\mG}^{(t+1)}_{p} \geq \tilde{\mG}^{(t)}_{p}$ for all $t \geq 0$ (Lemma~\ref{lem:gptpsd}). Furthermore, $\tilde{\mG}_p = \lim_{t \to \infty} \tilde{\mG}^{(t)}_{p}$ if the limit is defined.

We will simultaneously show that the limit above is defined and satisfies $\tilde{\mG}_p \leq \mD$.
Define the sequence $\mX_0 = \mD$ and $\mX_{t+1} = \tilde{F}_p(\mX_t)$ for $t \geq 0$. Clearly all the matrices in the sequence are positive semi-definite, and furthermore we claim that they satisfy $\mX_t \geq \mX_{t+1}$ for all $t$. The case $t = 0$ is immediate since $\mX_{1} = \tilde{F}_p(\mD) = \mPi F_p(\mD) \mPi = \mPi \mD \mPi = \mPi \mD \leq \mD = \mX_0$.
For $t > 0$ we use induction and the fact that $\tilde{F}_p$ is monotonous (Lemma~\ref{lem:FpFsmonotone}): if $\mX_{t} \geq \mX_{t+1}$, then $\mX_{t+1} = \tilde{F}_p(\mX_{t}) \geq \tilde{F}_p(\mX_{t+1}) = \mX_{t+2}$.
Thus, since $\tilde{\mG}^{(0)}_{p} = \mat{0} \leq \mD = \mX_0$, for all $t \geq 0$ we have $\tilde{\mG}^{(t)}_{p} \leq \mX_{t} \leq \mD$. This implies that the monotonously increasing sequence $\tilde{\mG}^{(t)}_{p}$ is bounded by $\mD$, and therefore its limit exists and is upper bounded by $\mD$.
\end{proof}

The above lemma will be enough to prove the desired upper bound on the norm of $\hat{f}$. On the other hand, we note that because $\tilde{A}$ is not a minimal WFA, the boundedness of $\hat{f}$ or the existence of the reachability Gramian $\tilde{\mG}_p$ do not immediately imply the existence of the observability Gramian $\tilde{\mG}_s$ for $\tilde{A}$; we will see in the next section that in fact this Gramian is also defined.

\begin{proof}[Proof of Theorem~\ref{thm:normfhat}]
By Lemma~\ref{lem:equivsvatrunc} we can work with $\tilde{A}$ instead of $\hat{A}$. Now, Lemma~\ref{lem:gramstildeA} shows the Gramian $\tilde{\mG}_p$ of $\tilde{A}$ is defined, so by Theorem~\ref{thm:gramiannorm} the function $\hat{f}$ computed by $\tilde{A}$ has finite $\ltwo$ norm. Furthermore, since $\mG_{A,p} \geq \mG_{\tilde{A},p}$, the expressions for the norm $\norm{f}_2$ in Theorem~\ref{thm:gramiannorm} imply that $\norm{f}_2^2 = \ainf^\top \mG_{A,p} \ainf \geq \ainf^\top \mG_{\tilde{A},p} \ainf = \norm{\hat{f}}_2^2$.
\end{proof}

\subsection{SVA Trucantion: Error Analysis}\label{sec:svaboundproof}

In this section we prove the bound on $\norm{f - \hat{f}}_2$ given in Theorem~\ref{thm:nicebound}, where $\hat{f}$ is the function computed by the SVA truncation of $f$ with $\hat{n}$ states. In fact, the bound will follow from an exact closed-form expression for the error $\norm{f - \hat{f}}_2$ given in terms of the Gramians of a WFA computing $\bar{f} = f - \hat{f}$.

We recall from last section the automaton $\tilde{A} = A \mPi$ with $n$ states computing $\hat{f}$, where $\mPi = \diag(\mI_{\hat{n}}, \mat{0})$. Now we proceed to combine $A$ and $\tilde{A}$ to obtain a WFA computing the difference $\bar{f} = f - \hat{f}$. The construction follows the same argument used to show that the difference of two rational functions is a rational functions, and yields the WFA $\bar{A} = \bwa$ with $2n$ states given by
\begin{align*}
\bazero &= \left[ \begin{array}{c} \azero \\ \tazero \end{array} \right] \enspace,  \\
\bainf &= \left[ \begin{array}{c} \ainf \\ -\tainf \end{array} \right] \enspace,
\\
\bA_\symbola &=  \left[ \begin{array}{cc} \A_\symbola & \mat{0} \\ \mat{0} &
\tA_\symbola \end{array} \right] = \diag(\A_\symbola,\tA_\symbola) \enspace.
\end{align*}
It is immediate to check from this constructions that $\bar{A}$ satisfies $f_{\bar{A}} = \bar{f}$. The following lemmas establish a few preliminary facts about $\bar{A}$.

\begin{lemma}
The observability Gramian $\tilde{\mG}_s$ of $\tilde{A}$ is defined.
\end{lemma}
\begin{proof}
Let $\H_{\hat{f}} = \tilde{\mP} \tilde{\mS}^\top$ be the factorization induced by $\tilde{A}$ and recall that $\tilde{\mG}_s = \tilde{\mS}^\top \tilde{\mS}$ if the corresponding inner products between the columns of $\tilde{\mS}$ are defined. Thus, to prove that $\tilde{\mG}_s$ is defined it suffices to show that all the columns of $\tilde{\mS}$ have finite $\ltwo$ norm, which is equivalent to showing that $\normf{\tilde{\mS}} < \infty$. Expanding this Frobenius norm we have
\begin{align}
\normf{\tilde{\mS}}^2 &=
\sum_{x \in \sstar} \norm{\tA_x \ainf}_2^2 \notag \\
&= \sum_{x \in \sstar} \ainf^\top \tA_x^\top \tA_x \ainf \notag \\
&= \sum_{x \in \sstar} \Tr(\tA_x \ainf \ainf^\top \tA_x^\top) \label{eqn:normst} \enspace,
\end{align}
where the last equality uses the cyclic property of the trace. Now note that using the SVA fixed-point equation $\mD = \ainf \ainf^\top + \sum_a \A_a \mD \A_a^\top$ we can rewrite any term in the infinite sum as
\begin{align*}
\Tr(\tA_x \ainf \ainf^\top \tA_x^\top) &=
\Tr\left(\tA_x \left(\mD - \sum_a \A_a \mD \A_a^\top\right) \tA_x^\top\right) \\
&= \Tr\left(\tA_x \left(\mD - \sum_a \A_a \mPi \mD \A_a^\top - \sum_a \A_a (\mI -\mPi) \mD \A_a^\top \right) \tA_x^\top\right) \enspace.
\end{align*}
Since $\mPi$ is idempotent and commutes with diagonal matrices, we have $\A_a \mPi \mD \A_a^\top = \A_a \mPi \mD \mPi \A_a^\top = \tA_a \mD \tA_a^\top$. Therefore, by linearity of the matrix trace, we can plug the last two observations into \eqref{eqn:normst} and get
\begin{align*}
\normf{\tilde{\mS}}^2 &=
\Tr(\mD)
- \sum_a \Tr (\tA_a \mD \tA_a^\top)
- \sum_a \Tr(\A_a (\mI - \mPi) \mD \A_a^\top) \\
&\;\; + \sum_{x \in \Sigma^+} \Tr(\tA_x \mD \tA_x^\top)
- \sum_{x \in \Sigma^+} \sum_a \Tr(\tA_x \tA_a \mD \tA_a^\top \tA_x^\top) \\
&\;\;
- \sum_{x \in \Sigma^+} \sum_a \Tr(\tA_x \A_a (\mI - \mPi) \mD \A_a^\top \tA_x^\top) \enspace.
\end{align*}
By aggregating terms we see that all terms of the form $ \Tr(\tA_x \mD \tA_x^\top)$ for $x \in \Sigma^+$ cancel and finally get
\begin{equation}
\normf{\tilde{\mS}}^2 = \Tr(\mD) - \sum_{x \in \sstar} \sum_a \Tr(\tA_x \A_a (\mI - \mPi) \mD \A_a^\top \tA_x^\top) \leq \Tr(\mD) \enspace,
\end{equation}
where we used that $\tA_x \A_a (\mI - \mPi) \mD \A_a^\top \tA_x^\top \geq \mat{0}$ and the trace of a positive semi-definite matrix is always non-negative.
\end{proof}

\begin{lemma}
Let $\H_f = \mP \mS^\top$ be the factorization induced by the SVA $A$, and $\H_{\hat{f}} = \tilde{\mP} \tilde{\mS}^\top$ be the factorization induced by $\tilde{A}$. Then the WFA $\bar{A}$ computing $\bar{f} = f - \hat{f}$ induces the factorization $\H_{\bar{f}} = \bar{\mP} \bar{\mS}^\top$ with $\bar{\mP} = [ \mP \; \tilde{\mP}]$ and $\bar{\mS} = [\mS \; -\tilde{\mS}]$. Furthermore, the Gramians $\bar{\mG}_p$ and $\bar{\mG}_s$ are defined and can be written as
\begin{align*}
\bar{\mG}_p &= \bar{\mP}^\top \bar{\mP} = \left[
\begin{array}{cc}
\mG_p & \mP^\top \tilde{\mP} \\
\tilde{\mP}^\top \mP & \tilde{\mG}_p
\end{array}
\right] \enspace, \\
\bar{\mG}_s &= \bar{\mS}^\top \bar{\mS} = \left[
\begin{array}{cc}
\mG_s & - \mS^\top \tilde{\mS} \\
- \tilde{\mS}^\top \mS & \tilde{\mG}_s
\end{array}
\right] \enspace.
\end{align*}
\end{lemma}
\begin{proof}
The structure of $\bar{\mP}$, $\bar{\mS}$, $\bar{\mG}_p$, and $\bar{\mG}_s$ follow from a straightforward computation. That these Gramians are defined follows from noting that because all the Gramians of $A$ and $\tilde{A}$ are defined, then all the columns of $\mP$, $\mS$, $\tilde{\mP}$, and $\tilde{\mS}$ have finite $\ltwo$ norm.
\end{proof}

Now we are ready to prove the main result of this section giving an exact closed-form expression for the $\ltwo$ distance between $f$ and $\hat{f}$.

\begin{theorem}\label{thm:exactbound}
For any truncation size $1 \leq \hat{n} < n$ we have
\begin{equation}
\norm{f - \hat{f}}_2^2 =
\Tr \left( \mD^{1/2} (\mI - \mPi) (\tilde{\mS} \mS^\top + \mS \tilde{\mS}^\top - \tilde{\mS} \tilde{\mS}^\top) (\mI - \mPi) \mD^{1/2} \right) \label{eqn:exacterror} \enspace.
\end{equation}
\end{theorem}
\begin{proof}
Recall from Theorem~\ref{thm:gramiannorm} that $\norm{\bar{f}}_2^2 = \bazero^\top \bar{\mG}_s \bazero = \Tr(\bazero^\top \bar{\mS}^\top \bar{\mS} \bazero) = \Tr(\bar{\mS} \bazero \bazero^\top \bar{\mS}^\top)$, where the last equality follows from a standard property of the trace. Note that by construction of $\bar{A}$ we have $\bazero^\top = \azero^\top [ \mI \; \mPi ]$, which when plugged in the previous equation yields
\begin{equation}
\norm{\bar{f}}_2^2 = \Tr\left(\bar{\mS} \left[\begin{array}{c} \mI \\ \mPi \end{array} \right] \azero \azero^\top [\mI \; \mPi] \bar{\mS}^\top\right) \enspace. \label{eqn:trace1}
\end{equation}

Recall that $A$ is an SVA, and therefore the fixed-point equation \eqref{eqn:fpalpha} applied to $A$ yields $\azero \azero^\top = \mD - \sum_{\symbola} \A_\symbola^\top \mD \A_\symbola$. When combined with \eqref{eqn:trace1} we obtain, by linearity of the trace:
\begin{align}
\Tr\left(\bar{\mS} \left[\begin{array}{c} \mI\\ \mPi \end{array} \right] \azero \azero^\top [\mI \; \mPi] \bar{\mS}^\top \right)
&=
\Tr\left(\bar{\mS} \left[\begin{array}{c} \mI \\ \mPi \end{array} \right] \mD [\mI \; \mPi] \bar{\mS}^\top \right) \label{eqn:step1}
\\
&-
\Tr\left(\bar{\mS} \left[\begin{array}{c} \mI \\ \mPi \end{array} \right]  
\left(\sum_\symbola \A_\symbola^\top \mD \A_\symbola \right)
[\mI \; \mPi] \bar{\mS}^\top \right) \enspace. \notag
\end{align}
Using that $[\mI \; \mPi] = [\mI \; \mI] - [\mat{0} \; \mI - \mPi]$, we decompose the first term as:
\begin{align}
& \Tr\left(\bar{\mS} \left[\begin{array}{c} \mI \\ \mI \end{array} \right] \mD [\mI \; \mI] \bar{\mS}^\top \right)
+
\Tr\left(\bar{\mS} \left[\begin{array}{c} \mat{0} \\ \mI - \mPi \end{array} \right] \mD [\mat{0} \; \mI - \mPi] \bar{\mS}^\top \right)
\label{eqn:term1}
\\
-
&\Tr\left(\bar{\mS} \left[\begin{array}{c} \mI \\ \mI \end{array} \right] \mD [\mat{0} \; \mI - \mPi] \bar{\mS}^\top \right)
-
\Tr\left(\bar{\mS} \left[\begin{array}{c} \mat{0} \\ \mI - \mPi \end{array} \right] \mD [\mI \; \mI] \bar{\mS}^\top \right) \enspace. \notag
\end{align}

We now proceed to bound the sum of the last three terms in this expression. Note in the first place that each of these terms is of the form $\Tr(\mM \mD \mN) = \Tr(\mD^{1/2} \mN \mM \mD^{1/2}) = \TrD(\mN \mM)$, where in the last step we just introduced a bit of convenient notation. Furthermore, recall that by definition of $\bar{A}$ we have $\bar{\mS} = [\mS \; -\tilde{\mS}]$. With these observations we obtain the following three equations:
\begin{align}
\Tr\left(\bar{\mS} \left[\begin{array}{c} \mat{0} \\ \mI - \mPi \end{array} \right] \mD [\mat{0} \; \mI - \mPi] \bar{\mS}^\top \right)
&=
\TrD\left((\mI - \mPi) \tilde{\mS}^\top \tilde{\mS} (\mI - \mPi) \right) \enspace,
\\
- \Tr\left(\bar{\mS} \left[\begin{array}{c} \mI \\ \mI \end{array} \right] \mD [\mat{0} \; \mI - \mPi] \bar{\mS}^\top \right)
&=
\TrD\left((\mI - \mPi) \tilde{\mS}^\top (\mS - \tilde{\mS}) \right) \enspace,
\\
- \Tr\left(\bar{\mS} \left[\begin{array}{c} \mat{0} \\ \mI - \mPi \end{array} \right] \mD [\mI \; \mI] \bar{\mS}^\top \right)
&=
\TrD\left( (\mS^\top - \tilde{\mS}^\top) \tilde{\mS} (\mI - \mPi) \right)
\enspace.
\end{align}
By observing that we have $\TrD((\mI - \mPi) \mM) = \TrD(\mM (\mI - \mPi)) = \TrD((\mI - \mPi) \mM (\mI - \mPi))$ for any square matrix $\mM$, we conclude that the sum of the last three terms in \eqref{eqn:term1} equals
\begin{equation}
\TrD\left( (\mI - \mPi) (\tilde{\mS}^\top \mS + \mS^\top \tilde{\mS} - \tilde{\mS}^\top \tilde{\mS}) (\mI - \mPi) \right) \label{eqn:bound1}
\enspace.
\end{equation}

To complete the proof of the equation we will now show that the sum of the remaining terms in \eqref{eqn:step1} vanish; that is:
\begin{equation}
\Tr\left(\bar{\mS} \left[\begin{array}{c} \mI \\ \mI \end{array} \right] \mD [\mI \; \mI] \bar{\mS}^\top \right)
-
\Tr\left(\bar{\mS} \left[\begin{array}{c} \mI \\ \mPi \end{array} \right]  
\left(\sum_\symbola \A_\symbola^\top \mD \A_\symbola \right)
[\mI \; \mPi] \bar{\mS}^\top \right) = 0 \enspace. \label{eqn:last}
\end{equation}
We start by noting the following identity:
\begin{equation}
\A_\symbola [ \mI \; \mPi ] = [ \A_\symbola \, \hA_\symbola ] = [\mI \; \mI] \bA_\symbola \enspace.
\end{equation}
Therefore, using the fixed-point equation \eqref{eqn:fpbeta} we see that
\begin{align*}
\Tr\left(\bar{\mS} \left[\begin{array}{c} \mI \\ \mPi \end{array} \right]  
\left(\sum_\symbola \A_\symbola^\top \mD \A_\symbola \right)
[\mI \; \mPi] \bar{\mS}^\top \right)
&=
\Tr\left(\bar{\mS} \left(\sum_\symbola \bA_\symbola^\top \left[\begin{array}{c} \mI \\ \mI \end{array} \right] \mD [\mI \; \mI] \bA_\symbola \right)
\bar{\mS}^\top \right)
\\
&=
\TrD\left([\mI \; \mI] \left(\sum_\symbola \bA_\symbola \bar{\mS}^\top \bar{\mS} \bA_\symbola^\top  \right) \left[\begin{array}{c} \mI \\ \mI \end{array} \right] \right)
\\
&=
\TrD\left([\mI \; \mI] \left(\bar{\mS}^\top \bar{\mS} - \bainf \bainf^\top \right) \left[\begin{array}{c} \mI \\ \mI \end{array} \right] \right) \enspace.
\end{align*}
Now \eqref{eqn:last} follows from simply observing that by the construction of $\bar{A}$ we have $[\mI \; \mI] \bainf = 0$.
\end{proof}

Finally we can show how the bound in Theorem~\ref{thm:nicebound} follows directly from the exact expression for the error obtained in Theorem~\ref{thm:exactbound}.

\begin{proof}[Proof of Theorem~\ref{thm:nicebound}]
We start noting that \eqref{eqn:exacterror} can be rewritten as
\begin{equation}
\TrD\left( (\mI - \mPi) \mS^\top \mS (\mI - \mPi) \right)
- \TrD\left(
(\mI - \mPi)
\left[ \mI \; \mI \right]
\bar{\mS}^\top \bar{\mS}
\left[ \begin{array}{c} \mI \\ \mI \end{array} \right]
(\mI - \mPi) \right)
\enspace.
\end{equation}
Note that the second term has the form $\Tr(\mM \mM^\top)$ and therefore is non-negative. Using that $A$ is an SVA, we see that the first term is
\begin{equation}
\TrD\left( (\mI - \mPi) \mS^\top \mS (\mI - \mPi) \right) =
\Tr\left( \mD^{1/2} (\mI - \mPi) \mD (\mI - \mPi) \mD^{1/2} \right) =
\sum_{i = \hat{n}+1}^{n} \sval_i^2 \enspace.
\end{equation}
Thus, it follows from the last two observations that \eqref{eqn:exacterror} is at most $\sum_{i = \hat{n}+1}^{n} \sval_i^2$.
\end{proof}

\section{Related Work}\label{sec:related}

In this section we provide wider context for our work by relating it to
recent developments in machine learning and to well established results in
the theory of linear dynamical systems. 

Spectral techniques for learning weighted automata and other latent
variable models have recently drawn a lot of attention in the machine
learning community.
Following the significant milestone papers~\citep{hsu09,denis}, in which an
efficient spectral algorithm for learning hidden Markov models (HMM) and
stochastic rational languages was given, the field has grown very rapidly.
The original algorithm, which is based on singular value decompositions of
finite sub-blocks of Hankel matrices, has been extended to reduced-rank
HMMs~\citep{siddiqi10}, predictive state representations
(PSR)~\citep{Boots:2009}, finite-state transducers~\citep{fst,bailly2013fst},
and many other classes of functions on strings
\citep{bailly11,nips12,recasens2013spectral}.
Although each of these papers works with slightly different problems and
analysis techniques, the key ingredient turns out to be always the same:
parametrize the target model as a WFA and learn
this WFA from the SVD of a finite sub-block of its Hankel
matrix~\citep{mlj13spectral}. 
Therefore, it is possible (and desirable) to study all these learning algorithms
from the point of view of rational series, which are exactly the class of
real-valued functions on strings that can be computed by WFA.

The appeal of spectral learning techniques comes from their
computational superiority when compared to iterative algorithms like
Expectation--Maximization (EM)~\citep{dempster}.
Another very attractive property of spectral methods is the possibility of
proving rigorous statistical guarantees about the learned automaton.
For example, under a realizability assumption, these methods are known to be
consistent and amenable to finite-sample analysis in the PAC sense~\citep{hsu09}.
An important detail is that, in addition to realizability, these results work
under the assumption that the user correctly guesses the number of latent
states of the target distribution.
Though this is not a real caveat when it comes to using these algorithms in
practice -- the optimal number of states can be identified using a model
selection procedure~\citep{icml2014balle} -- it is one of the barriers in
extending the statistical analysis of spectral methods to the non-realizable
setting.

Tackling the non-realizability question requires, as a special case, dealing
with the situation in which data is generated from a WFA with $n$ states and the
learning algorithm is asked to produce a WFA with $\hat{n} < n$ states.
This case is already a non-trivial problem which -- barring the noisiness
introduced by estimating the Hankel matrix from observed data --
can in fact be interpreted as an approximate minimization of WFA. 
From this point of view, we believe our results provide the fundamental tools necessary for
addressing important problems in the theory of learning weighted automata,
including the robust statistical analysis of spectral learning algorithms.

A connection between spectral learning algorithms and approximate
minimization for a small class of hidden Markov models was considered
in~\citep{kulesza2014low}.  This paper also presents a
theoretical result bounding the error between the original and minimized
HMM in terms of the total variation distance. The bounds in this paper are
incomparable to ours.  However, in a follow-up work \citep{kulesza2015low},
published concurrently with our original paper on SVA
\citep{balle2015canonical}, a problem similar to the one considered here is
addressed, albeit different methods are used and the results are less
general that our approximate minimization method.
Another paper on which the issue of approximate minimization of weighted
automata is considered in a tangential manner is \citep{14KW-ICALP}.
In this case the authors again focus on an $\lone$-like accuracy measure to
compare two automata: an original one, and another one obtained by removing
transitions with small weights occurring during an exact minimization procedure.
Though the removal operation is introduced as a means of obtaining a numerically
stable minimization algorithm, the paper also presents some experiments
exploring the effect of removing transitions with larger weights.
With the exception of these timid results, the problem of approximate
minimization for general WFA remained largely unstudied before our paper. 

However, the case of an alphabet with one symbol, $|\Sigma| = 1$, has been
thoroughly studied from multiple points of view.  In the control theory
literature several methods have been proposed for approximate minimization
of time-invariant linear dynamical systems under the names of model
reduction, truncation, and approximation; see \citep{antoulas2005approximation} for a comprehensive presentation. One possible approach
to the model reduction problem is to consider so-called balanced
realizations of a linear dynamical system and apply a convenient truncation
method to the balanced realization to obtain a smaller system
\citep{enns1984model}. In the one symbol case, the connection with weighted automata
arises from observing that the impulse response of a time-invariant linear
dynamical system can be parametrized as a weighted automata with one letter
(and possibly vector-valued outputs for multiple-input multiple-output
systems) \citep{antoulas2005approximation}. From this point of view, the canonical form for
weighted automata given by our SVA can be interpreted as a generalization of balanced
realizations to the case where the alphabet has two or more letters.

The study of model reduction techniques in the one symbol case can also be
connected to sophisticated ideas in the study of approximations for Hankel
operators in the functional and complex analysis literatures; see e.g.\ \citep{peller2012hankel} for a comprehensive treatment of the theory of Hankel operators. In
the same way we do in Section~\ref{sec:ratfuns}, when the alphabet $\Sigma$ has
only one symbol the Hankel matrix of a rational function yields a linear
operator between Hilbert or Bannach spaces of sequences. The spectral
properties of these Hankel operators have been thoroughly studied. For
example, deep connections to the theory of complex function on the unit
disk and Fourier analysis have been uncovered \citep{nikol2012treatise}.
Along these lines one finds the celebrated AAK theorem characterizing optimal approximations of Hankel operators by Hankel operators of bounded rank \citep{adamyan1971analytic}. This theorem has been widely exploited in control theory to provide alternative approaches to balanced realizations for model reduction, thus providing a link between the abstract setting of Hankel operators and the concrete problem of approximating of linear dynamical systems \citep{glover1984all} (see also \citep{fuhrmann2011polynomial}).
One of the fundamental ideas in this line of work is realizing that for
$|\Sigma| = 1$ the free monoid $\sstar$ can be identified with the natural
numbers $\N$, which can be canonically embedded in the \emph{abelian} group
$\Z$.
Unfortunately for us, this approach cannot be directly generalized to the case
$|\Sigma| > 1$ because in this case the corresponding embedding yields a free
non-abelian group, and standard Fourier analysis on those groups is not
available. Although some recent attempts have been made to extend some of the results
about Hankel operators to the non-commutative case using methods from functional analysis \citep{popescu2003multivariable}, this theory is still largely underdeveloped, and the few existing results can only be obtained via non-constructive arguments.

\section{Conclusion and Future Work}\label{sec:conclusion}

In the present paper we have given a new approximate minimization technique
based on spectral theory ideas.  The essential point was to use the
singular value decomposition to decide how to truncate the original
automaton without losing too much accuracy.  We have given quantitative
bounds on how close the approximate machine is to the original.

One crucial aspect that we have not addressed is the question of
constructing the \emph{best possible} approximation given a bound on the
size of the state space or, equivalently, the dimension of the vector space
on which the machine is defined. In the one-letter case,
sophisticated results from the theory of
Hankel operators~\citep{adamyan1971analytic,peller2012hankel} provide a satisfactory answer to this problem.  However,
extending this to the multiple-letter case means extending an already deep
and difficult theory to the non-commutative case.  Nevertheless, it remains
an exciting challenge.

A different approach is to change the approximation measure from $\ltwo$ to a more natural metric between WFA.
In recent work~\citep{Balle17} we developed a metric to measure the distance between
WFAs based on bisimulation. This metric has interesting properties , but unfortunately it is hard to compute for the
present type of approximation.  Nevertheless, it might be fruitful to
explore approximation schemes based on approximate bisimulation as has been
done for some types of Markov processes~\citep{Desharnais03}.  It would be
interesting to compare the quality of such approximation schemes with the
present one.

\section*{Acknowledgements}

We thank Fran\c{c}ois Denis for sharing the
proof of Lemma~\ref{lem:denis} with us,
Guillaume Rabusseau for useful discussions, and an anonymous referee for suggesting improvements to the presentation of the paper.

\bibliographystyle{abbrvnat}
\bibliography{paper}

\end{document}